\documentclass{amsart}
\pdfoutput=1

\usepackage{amsmath}
\usepackage{amssymb}
\usepackage{amsfonts}
\usepackage{graphicx}
\usepackage{texdraw}
\usepackage{pmboxdraw}
\usepackage{graphpap}
\usepackage{wasysym}
\usepackage{enumitem}
\usepackage{color}
\usepackage[OT2,T1]{fontenc}
\usepackage{verbatim} 
\usepackage{multirow}
\usepackage{mathtools}
\usepackage[normalem]{ulem}
\usepackage{cite}
\usepackage[colorlinks=true, citecolor=blue, filecolor=black, linkcolor=black, urlcolor=black]{hyperref}
\usepackage{bm}
\usepackage{bbm}
\usepackage{todonotes}
\usepackage{kantlipsum}
\allowdisplaybreaks

\theoremstyle{plain}
\newtheorem{thm}{Theorem}[section]

\newtheorem{cor}[thm]{Corollary}
\newtheorem{lem}[thm]{Lemma}

\newtheorem{defi}[thm]{Definition}

\theoremstyle{remark}
\newtheorem{rem}{Remark}

\newtheorem{ex}{Example}

\newcommand{\R}{{\mathbb R}}

\newcommand{\re}{\operatorname{Re}}
\newcommand{\im}{\operatorname{Im}}

\renewcommand{\d}{{\partial}}

\def\P{\mathbf{P}}

\def\FF{\mathcal{F}}

\def\C{\mathbb{C}}

\def\H{\mathbb{H}}
\def\P{\mathbf{P}}
\def\R{\mathbb{R}}

\newcommand{\totalN}{n}

\numberwithin{equation}{section}

\newcommand{\RN}[1]{%
	\textup{\uppercase\expandafter{\romannumeral#1}}%
}

\begin{document}
\title[Partition functions of the spherical Coulomb gases]{Free energy of spherical Coulomb gases with point charges}

\author{Sung-Soo Byun}
\address{Department of Mathematical Sciences and Research Institute of Mathematics, Seoul National University, Seoul 151-747, Republic of Korea}
\email{sungsoobyun@snu.ac.kr}

\author{Nam-Gyu Kang}
\address{School of Mathematics and June E Huh Center for Mathematical Challenges, Korea Institute for Advanced Study, 85 Hoegiro, Dongdaemun-gu, Seoul 02455, Republic of Korea}
\email{namgyu@kias.re.kr}

\author{Seong-Mi Seo}
\address{Department of Mathematics, Chungnam National University, 99 Daehak-ro, Yuseong-gu, Daejeon 34134, Republic of Korea.}
\email{smseo@cnu.ac.kr}

\author{Meng Yang}
\address{Department of Mathematics, School of Sciences,  Great Bay University, Dongguan, 523000, China}
\email{my@gbu.edu.cn}
 

\begin{abstract}
We consider two-dimensional Coulomb gases on the Riemann sphere with determinantal or Pfaffian structures, under external potentials that are invariant under rotations around the axis connecting the north and south poles, and with microscopic point charges inserted at the poles. These models can be interpreted as Coulomb gases on the complex plane with weakly confining potentials, where the associated droplet is the entire complex plane. For these models, we derive precise asymptotic expansions of the free energies, including the constant terms.
\end{abstract}

\maketitle

\section{Introduction and main results}

In the theory of two-dimensional Coulomb gases \cite{Fo10}, the asymptotic expansion of the free energy, as the system size $N$ increases, is a central and longstanding problem. In this expansion, the coefficients are believed to provide essential information about the model, including potential theoretical quantities, topological characteristics, and conformal geometric properties \cite{CFTW15,ZW06,JMP94,TF99}. In addition, the expansion of the free energy plays a crucial role in establishing fluctuations and the Gaussian free field convergence of Coulomb gases, see e.g. \cite{Se24} for more details. 

There has been remarkable progress in obtaining such expansions for general Coulomb gases, with the first three terms up to the $O(N)$ term, including the energy and entropy terms, being derived in \cite{LS17}, see also \cite{BBNY19,Se23}. On the one hand, early conjectures \cite{CFTW15,ZW06,JMP94,TF99} suggest that the expansion beyond the entropy term begins to depend on the topological properties of the Coulomb gas. For general inverse temperature $\beta > 0$, obtaining a complete expansion still seems a highly challenging problem. On the other hand, for a certain class of ensembles with $\beta = 2$, where one can utilise their notable integrable structures, combining cumulative knowledge from orthogonal polynomial and random matrix theory, has recently led to precise expansions up to the conjectural $O(1)$-terms \cite{BKS23,ACC23,BSY24,BFL25}.

Recent work has primarily focused on Coulomb gases on the plane. However, as discussed in \cite{JMP94} with a prototypical example, models on the Riemann sphere are expected to exhibit different behaviour in their asymptotic expansions. In this work, we aim to systematically investigate and contribute to the free energy expansion of such spherical Coulomb gases. Furthermore, by inserting microscopic point charges, we investigate their effects on the expansions, which go beyond the conjectural forms in the previous literature.

\subsection{Spherical Coulomb gases and weakly confining potential}

We first introduce the models, the spherical Coulomb gases. 
Let $\mathbb{S}= \{ x \in \mathbb{R}^3: \| x \| =1 \}$ be the Riemann sphere equipped with the Euclidean metric $\| \cdot \|$ in $\R^3$. 
We investigate two different types of Coulomb gases $\boldsymbol{x}=\{x_j\}_{j=1}^N$ on $\mathbb{S}$, known as the determinantal and Pfaffian Coulomb gases due to their respective integrable structures. Their joint probability distribution functions are of the form
\begin{align}
 \label{Gibbs cplx sphere}
d\boldsymbol{\mathcal{P}}_{N}^\C(z)&=\frac{1}{ \mathcal{Z}_{N }^{ \mathbb{C} }( Q_{ \mathbb{S} } )  } \prod_{j>k=1}^N \| x_j-x_k \|^{2} \prod_{j=1}^{N}  e^{-N Q_{ \mathbb{S} }(x_j) }  \,dA_{ \mathbb S }(x_j),
\\
d\boldsymbol{\mathcal{P}}_{N}^\H(z)&=\frac{1}{   \mathcal{Z}_{N }^{ \mathbb{H} }( Q_{ \mathbb{S} } ) } \prod_{j>k=1}^N \| x_j-x_k \|^{2} \| x_j-\overline{x}_k \|^{2} \prod_{j=1}^{N} \|x_j-\overline{x}_j\|^2  e^{-2N Q_{ \mathbb{S} }(x_j) }  \,dA_{ \mathbb S }(x_j),  \label{Gibbs symplectic sphere}
\end{align}
where $dA_{ \mathbb S }$ is the area measure on the sphere, and $\mathcal{Z}_{N }^{ \mathbb{C} }( Q_{ \mathbb{S} } )$ and $ \mathcal{Z}_{N }^{ \mathbb{H} }( Q_{ \mathbb{S} } )$ are partition functions that make \eqref{Gibbs cplx sphere} and \eqref{Gibbs symplectic sphere} probability measures. 
Here, $\overline{u}=(u_1,u_2,-u_3)$ for $u= (u_1,u_2,u_3)$, and $Q_\mathbb{S} : \mathbb{S} \to \R$ is the external potential, which may depend on $N$. 
The prototypical example of the potential is of the form
\begin{equation} \label{Q point charges on the sphere}
Q_{  \mathbb{S} }^{ \rm sp }(x) := -\frac{2\alpha}{N} \log \| x- {\rm{N}} \|   -\frac{2c}{N} \log \| x- {\rm{S}} \|,  \qquad \begin{cases}
{\rm{N}}=(0,0,1),
\smallskip 
\\
{\rm{S}}=(0,0,-1), 
\end{cases} 
\end{equation}
where $\alpha, c \ge 0$ are point charges. 
Notice that the potentials can be defined up to additive constants.  Namely, if $Q_{ \mathbb{S}}$ is replaced by $Q_{ \mathbb{S}}+\mathfrak{c}$ for some constant $\mathfrak{c}$, the models \eqref{Gibbs cplx sphere} and \eqref{Gibbs symplectic sphere} remain unchanged.  

It is convenient to consider equivalent models in the complex plane. The determinantal and Pfaffian Coulomb gas ensembles $\boldsymbol{z}=\{z_j\}_{j=1}^N$ in the complex plane $z_j \in \mathbb{C}$ are given by
\begin{align}
 \label{Gibbs cplx}
d\P_{N}^\C( \boldsymbol{z} )&=\frac{1}{Z_{N}^\C(Q) } \prod_{j>k=1}^N |z_j-z_k|^{2} \prod_{j=1}^{N}  e^{-N Q(z_j) }  \,dA(z_j),
\\
d\P_{N}^\H(z)&=\frac{1}{Z_{N}^\H(Q) } \prod_{j>k=1}^N |z_j-z_k|^{2} |z_j-\bar{z}_k|^2  \prod_{j=1}^{N}|z_j-\bar{z}_j|^2  e^{-2N Q(z_j) }  \,dA(z_j),  \label{Gibbs symplectic}
\end{align}
where $dA(z)=d^2z/\pi$ is the area measure. 
Here, the potential $Q_{ \mathbb S }$ is related to the potential $Q$ as 
\begin{equation} \label{rel btw Q S and Q}
Q_\mathbb{S}= Q \circ \phi - \frac{N+1}{N} \log (1+|\phi|^2), 
\end{equation}
up to an additive constant, where $\phi: \mathbb{S} \to \C \cup \{ \infty \}$ is the stereographic projection 
\begin{equation} \label{def of stereographic proj}
\phi(u_1,u_2,u_3) = 
\begin{cases}
\dfrac{ u_1+iu_2 }{ 1-u_3 }, &\textup{for } (u_1,u_2,u_3) \in \mathbb{S}^2, u_3 \not=1,
\smallskip 
\\
\infty, &\textup{for }(u_1,u_2,u_3)={\rm{N}}.  
\end{cases}
\end{equation}
The prototypical example \eqref{Q point charges on the sphere} then reads as
\begin{equation} \label{potential spherical ensemble N}
Q^{ \rm sp }(z):= \frac{ N+\alpha+c+1 }{N} \log(1+|z|^{2})-\frac{ 2c }{N} \log |z|. 
\end{equation} 
In this context, it is crucial that the potential is $N$-dependent. We also write 
\begin{equation} \label{potential spherical ensemble infty}
Q^{ \rm sp }_\infty(z):= \log(1+|z|^2) 
\end{equation}
for the large-$N$ limit of \eqref{potential spherical ensemble N}. 
The ensembles \eqref{Gibbs cplx} and \eqref{Gibbs symplectic} with the potential $Q^{ \rm sp }$ correspond to the eigenvalues of complex \cite{Kr09,FF11,FK09} and symplectic \cite{FM12,May13,MP17,BF23a} spherical induced ensembles, respectively.  
See \cite[Sections 2.5 and 11.3]{BF24} for comprehensive reviews.

The partition functions $\mathcal{Z}_{N}^{\mathbb{C}}(Q_{\mathbb{S}})$ and $\mathcal{Z}_{N}^{\mathbb{H}}(Q_{\mathbb{S}})$ in \eqref{Gibbs cplx sphere} and \eqref{Gibbs symplectic sphere} are related to $Z_{N}^{\mathbb{C}}(Q)$ and $Z_{N}^{\mathbb{H}}(Q)$ via the following property.

\begin{lem} \label{Lem_Z sphere plane relation} We have 
\begin{equation}
\mathcal{Z}_{ N }^{ \mathbb{C} }( Q_{\mathbb{S}}  )  = 2^{N(N-1)} Z_N^{\mathbb{C} } (Q), \qquad  \mathcal{Z}_{ N }^{ \mathbb{H} }( Q_{\mathbb{S}}  )  = 2^{2N^2} Z_N^{\mathbb{H} } (Q).
\end{equation}
\end{lem}

This lemma easily follows by the change of variables under the stereographic projections, see Subsection~\ref{Subsection_stereographic}. 


\medskip 

In general, there are two equivalent but slightly different perspectives on defining the Coulomb gases. 
The first approach is to consider a given potential $Q$ and define the model according to \eqref{Gibbs cplx} and \eqref{Gibbs symplectic}. Together with the classical equilibrium convergence \cite{Se24}, the geometric properties of the models then follow from logarithmic potential theory \cite{ST97}. Specifically, the system tends to be distributed according to Frostman's equilibrium measure associated with the potential \( Q \). The limiting support of the system consequently depends on the potential \( Q \). This approach is straightforward for defining the model; however, it is generally difficult to determine the geometry of the droplet a priori.
The second approach is somewhat the opposite. We begin with a certain probability measure supported on the complex plane, which coincides with the equilibrium measure. Using this measure, we then define the external potential via a logarithmic potential. While this approach is somewhat less explicit in defining the external potential, it has the advantage of allowing us to determine the geometry of the droplet a priori, as it is given by the support of the initial probability measure.

For our purpose of investigating Coulomb gases supported on the whole Riemann sphere, we take this second approach. This is closely related to the notion of weakly confining potential or the planar jellium model.
We first collect our assumptions on the background measure. 

\begin{defi}[\textbf{Background measure}] \label{Def_measure mu}
Let $\mu$ be an $N$-independent probability measure supported on the whole complex plane. We assume the following.
\begin{itemize}
    \item $\mu$ is absolutely continuous with respect to the area measure:
\begin{equation}
d\mu(z) = {\boldsymbol\rho} (z) \, dA(z),
\end{equation}
where $\boldsymbol\rho : \C \to \R_{ > 0 }$ is smooth. 
    \smallskip 
    \item $\boldsymbol\rho$ is radially symmetric with $\boldsymbol\rho(z) = \rho(|z|)$, where $\rho: \R_{ \ge 0 } \to \R_{ > 0 }$.  
    \smallskip 
    \item The pull-back of the density on the Riemann sphere 
\begin{equation}
\boldsymbol\rho_{ \mathbb{S} }(x)= (1+|\phi(x)|^2)^2 \boldsymbol\rho(\phi(x)) 
\end{equation}
does not vanish at the north pole $\rm N$. 
Equivalently, there exists a positive constant $A>0$ such that 
\begin{equation}\label{rho:growth}
    \boldsymbol\rho(z) \sim \frac{A}{|z|^4}, \qquad z\to \infty. 
\end{equation}
We write 
    \begin{equation} \label{def of rho tilde}
\tilde{\rho}(s) := s^{-4}\rho(s^{-1}), \qquad s>0. 
\end{equation} 
Then $\tilde{\rho}$ has the unique continuous extension to $\R_{\ge0}$ and $\tilde{\rho}(0) = A.$ We further assume that $\tilde{\rho}$ is smooth at $0$. 
\end{itemize}
\end{defi}


For a probability measure $\mu$ supported on the complex plane, we write  
\begin{equation} \label{def of logarithmic potential}
U_\mu(z) = \int \log\frac1{|z-w|}\, d\mu(w)
\end{equation}
for the associated logarithmic potential. 
For a given probability measure $\mu$ and positive real numbers $\alpha,c\ge 0$, we define the external potential 
\begin{equation} \label{def of external potential}
Q(z) \equiv Q_{N,\alpha,c}(z) := -2\frac{\totalN}{N} U_\mu(z)-\frac{2c}{N}\log|z|, \qquad \totalN:=N+\alpha+c+1.  
\end{equation}
Let us also write 
\begin{equation} \label{def of Q infty}
Q_\infty(z) := -2 U_\mu(z). 
\end{equation}

Note that since 
\begin{equation} \label{U mu growth}
U_\mu(z)= -\log|z|+ o(1) , \qquad z \to \infty,
\end{equation}
we have 
\begin{equation}\label{Q:growth}
Q_{N,\alpha,c}(z) \sim 2 \frac{N+\alpha+1}{N} \log|z|, \qquad z \to \infty. 
\end{equation}
This growth property is indeed a characteristic feature of the spherical ensemble. For a potential with a sufficiently strong growth condition \(Q(z) \gg 2 \log |z|\) as $z \to \infty$, it is well known that the associated droplet is a compact subset of the complex plane. In contrast, for a potential with the growth condition \eqref{Q:growth}, the confining energy is weak enough that the system \eqref{Gibbs cplx} lie in the whole complex plane as $N \to \infty$. A potential with the growth condition \eqref{Q:growth} is referred to as a weakly confining potential, see e.g. \cite{BGNW21} and references therein.
Let us also mention that for $c=0$, our definition of the models correspond to the definition of the planar jellium \cite{CGJ20}.  
Thus, our model can be regarded as the jellium model in the plane with total background charge $2(N+\alpha+c+1)$ and background distribution $\mu$ with an insertion of a point charge $c$ at the origin. 

\begin{ex}
The simplest and fundamental example is the case where $\boldsymbol\rho_{ \mathbb{S} }$ is uniform on the Riemann sphere. In the complex plane, it corresponds to the spherical density 
\begin{equation}
\boldsymbol\rho^{ \rm sp }_\infty(z) := \frac{1}{(1+|z|^2)^2}.
\end{equation}
In this case, we have 
$$
U_\mu(z)\Big|_{ \boldsymbol\rho= \boldsymbol\rho^{ \rm sp }_\infty } = - \frac12 \log(1+|z|^2) = - \frac12\, Q^{ \rm sp }_\infty(z),  
$$
where $Q^{ \rm sp }_\infty$ is given by \eqref{potential spherical ensemble infty}.
Then associated external potential is given by $Q^{ \rm sp }$ in \eqref{potential spherical ensemble N}. 
\end{ex}

\subsection{Main results}
 
For a given probability measure $\mu$ in Definition~\ref{Def_measure mu}, we write 
\begin{equation} \label{def of I mu}
 I[\mu] :=  \int \log|z-w| \,d\mu(z) \, d\mu(w)  = - \int U_{\mu} \,d\mu  
\end{equation}
for the (unweighted) logarithmic energy. By using \eqref{def of Q infty}, it can also be interpreted as a weighted logarithmic energy associated with the potential $Q_\infty$, namely, 
\begin{equation}
I[\mu] \equiv I_{ Q_\infty }[\mu] =  -\int \log|z-w| \,d\mu(z) \, d\mu(w)  + \int  Q_\infty \,d\mu.
\end{equation}
We also write 
\begin{equation}
 E[\mu] := \int \log \boldsymbol\rho \, d\mu
\end{equation} 
for the entropy.  

In the presence of local point charges, the expansions of the free energies are expressed in terms of the Barnes $G$-function, which is defined recursively by
\begin{equation} \label{def of Barnes G}
G(z+1)=\Gamma(z)G(z),\qquad G(1)=1, 
\end{equation}
see \cite[Section 5.17]{NIST}. 
We then have the following. 

\begin{thm}[\textbf{Free energy expansion of the determinantal Coulomb gas on the sphere}] \label{Thm_free energy det}
Suppose that a probability measure $\mu$ satisfies the assumptions in Definition~\ref{Def_measure mu}. Let $Q \equiv Q_{N,\alpha,c}$ be given in \eqref{def of external potential}.
Then as $N\to\infty$, the partition function $Z_N^{\C}(Q)$ in \eqref{Gibbs cplx} satisfies the asymptotic expansion 
\begin{equation}
 \log Z_N^{\C}(Q) = C_1 N^2 + C_2 N \log N + C_3 N +C_4 \log N+C_5 + O(N^{-\frac{1}{12}}(\log N)^{d}),
\end{equation}
for some constant $d>0$, where 
\begin{align}
C_1 &= - I [\mu], 
\\
C_2 &= \frac12 , 
\\
C_3 &= \frac{\log (2\pi)}{2} -1 - 2(\alpha+c+1) I[\mu] -\frac12 E[\mu] - 2c \,U_\mu(0),
\\
C_4 &= \frac{\alpha^2+c^2}{2} +\frac13, 
\\
\begin{split}
C_5 &= -(\alpha+c+1)^2 I[\mu] -(\alpha+c+1) \Big( \frac{1}{2} E[\mu] + 2c \, U_\mu(0) \Big)
\\
&\quad  +\frac{1}{2}\Big(c^2+c+\frac{1}{3}\Big)\log \rho(0)+\frac{1}{2}\Big(\alpha^2 + \alpha + \frac{1}{3}\Big)\log \tilde{\rho}(0) \\ 
    &\quad + \frac{1}{2}(\alpha+c+1)(\log (2\pi)-1)  +2\zeta'(-1)-\log \Big(G(c+1)G(\alpha+1)\Big)
    \\
     & \quad    - \frac{5}{12} - \frac{1}{6}\int_{0}^{\infty} \Big(\frac{\rho''(t)}{\rho(t)}-\frac{5}{4}\Big(\frac{\rho'(t)}{\rho(t)}\Big)^2\Big)t\,dt .
\end{split}
\end{align}
Here $\tilde{\rho}$ is given by \eqref{def of rho tilde}, $\zeta$ is the Riemann zeta function, and $G$ is the Barnes $G$-function. 
\end{thm} 

In contrast to the usual asymptotic formulas found in the literature (see also \eqref{Z expansion cplx gen} below), the formula for $C_3$ contains the energy term, while $C_5$ contains both the energy and entropy terms. This arises from the $N$-dependence of the potential \eqref{def of external potential}. Consequently, it is also natural to express the expansion in terms of the total background charge $n = N + \alpha + c + 1$ as
\begin{align}
\begin{split} \label{expansion ZN C nN}
 \log Z_N^{\C}(Q) & \sim - I[\mu]n^2 +\frac{1}{2} N\log N - \Big( \frac12 E[\mu]+2c\,U_\mu(0) \Big) n  +\Big( \frac{ \log (2\pi) }{2}-1 \Big) N +
        C_4 \log N \\
        &\quad +\frac{1}{2}\Big(c^2+c+\frac{1}{3}\Big)\log \rho(0)+\frac{1}{2}\Big(\alpha^2 + \alpha + \frac{1}{3}\Big)\log \tilde{\rho}(0) + \frac{1}{2}(\alpha+c+1)\log (2\pi) \\
        &\quad +2\zeta'(-1)-\log (G(c+1)G(\alpha+1)) - \frac{5}{12} - \frac{1}{6}\int_{0}^{\infty} \Big(\frac{\rho''(t)}{\rho(t)}-\frac{5}{4}\Big(\frac{\rho'(t)}{\rho(t)}\Big)^2\Big)t\,dt . 
\end{split}
\end{align}

\begin{rem} It is intuitively clear that the term $2c\,U_\mu(0)$ in $C_3$ arises from the microscopic insertion of a point charge at the origin. Notice however that the symmetric counterpart of \( 2c\,U_\mu(0) \) does not appear in \( C_3 \). This omission arises from our specific normalisation when defining the potential. To be more precise, we first write 
\begin{equation}
Q_\mathbb{S,\infty} := Q_\infty \circ \phi -   \log (1+|\phi|^2)
\end{equation}
for the large-$N$ counterpart of the potential on the sphere geometry. 
Note that by \eqref{U mu growth} and \eqref{def of Q infty}, 
$$
Q_\infty(z) = 2 \log |z| +o(1), \qquad Q_\mathbb{S,\infty}( {\rm{N}} ) =0.
$$ 
Therefore the seemingly asymmetric term $-2c \, U_\mu(0)$ can actually be written in a more symmetric manner
\begin{equation}
-2c \, U_\mu(0) = c\, Q_\mathbb{S,\infty}( {\rm{S}} ) = c\, Q_\mathbb{S,\infty}( {\rm{S}} )  + \alpha\, Q_\mathbb{S,\infty}( {\rm{N}} ) . 
\end{equation} 
\end{rem}

As a counterpart for the Pfaffian Coulomb gases, we have the following.

\begin{thm}[\textbf{Free energy expansion of the Pfaffian Coulomb gas on the sphere}] \label{Thm_free energy Pfaff}
Under the same assumptions in Theorem~\ref{Thm_free energy det}, as $N\to\infty$, the partition function $Z_N^{\mathbb{H}}(Q)$ in \eqref{Gibbs symplectic} satisfies 
\begin{equation}
 \log Z_N^{ \mathbb{H} }(Q) = D_1 N^2 + D_2 N \log N + D_3 N +D_4 \log N+D_5 + O(N^{-\frac{1}{12}}(\log N)^{d}),
\end{equation}
for some constant $d>0$, where 
\begin{align}
D_1 &= -2 I [\mu], 
\\
D_2 &= \frac12 , 
\\
D_3 &=  \frac{\log (4\pi)}{2} -1 - 4(\alpha+c+1) I[\mu] -\frac12 E[\mu] - (4c+1) U_\mu(0), 
\\
D_4 &= \alpha^2+\frac{\alpha}{2}+c^2 + \frac{c}{2}+\frac{5}{12},
\\
\begin{split}
D_5 &= -2(\alpha+c+1)^2 I[\mu] -(\alpha+c+1) \Big( \frac{1}{2} E[\mu] + (4c+1) \, U_\mu(0) \Big)
\\
&\quad + \Big(c^2+c+\frac{5}{24}\Big)\log \rho(0) + \Big(\alpha^2 + \alpha + \frac{5}{24}\Big)\log \tilde{\rho}(0)  \\ 
    &\quad + (\alpha+c+1)\Big(\log (2\pi)-\frac12 \Big)  +4\zeta'(-1)- \log\Big(G(c+1)G(c+\tfrac{3}{2})G(\alpha+1)G(\alpha+\tfrac{3}{2})\Big)
    \\
     & \quad   -\frac{5}{24}  - \frac{1}{12}\int_0^{\infty} \Big(\frac{\rho''(t)}{\rho(t)} - \frac{5}{4}\Big(\frac{\rho'(t)}{\rho(t)} \Big)^2 \Big)t\,dt. 
\end{split}
\end{align} 
Here $\zeta$ is the Riemann zeta function and $G$ is the Barnes $G$-function. 
\end{thm}

By using the scaling parameter $n$, we can again rewrite the expansion so that the energy and entropy terms appear only in the coefficients of $n^2$ and $n$, respectively:
\begin{align} \label{expansion ZN H nN}
\begin{split}
     \log Z_N^{ \mathbb{H} }(Q)  &\sim -2 I[\mu]\, n^2 + \frac12 N \log N -\Big( \frac12 E[\mu]+(4c+1) U_\mu(0) \Big)n  + \Big( \frac{\log (4\pi)}{2}-1\Big)N     + D_4 \log N \\
        &\quad + \Big(c^2+c+\frac{5}{24}\Big)\log \rho(0) + \Big(\alpha^2 + \alpha + \frac{5}{24}\Big)\log \tilde{\rho}(0)  + (\alpha+c+1 )\log (2\pi)  + 4\zeta'(-1) 
        \\
        &\quad - \log\Big(G(c+1)G(c+\tfrac{3}{2})G(\alpha+1)G(\alpha+\tfrac{3}{2})\Big) -\frac{5}{24}- \frac{1}{12}\int_0^{\infty} \Big(\frac{\rho''(t)}{\rho(t)} - \frac{5}{4}\Big(\frac{\rho'(t)}{\rho(t)} \Big)^2 \Big)t\,dt.  
\end{split}
\end{align} 

\begin{rem}
Notice that the term $D_3$ contains $(4c+1)U_\mu(0).$ As explained in \cite{BKS23}, due to the radial symmetry, $U_\mu(0)$ can also be rewritten as 
$
U_\mu(0) = - \int \log |z-\overline{z}| \,d\mu(z),  
$
which leads to 
\begin{equation}
(4c+1)U_\mu(0) =  4c \, U_\mu(0)  - \int \log |z-\overline{z}| \,d\mu(z). 
\end{equation}
Albeit equivalent, the expression on the right-hand side captures more of the model's statistical meaning, as the first term arises from the insertion of a point charge, while the second term reflects the local repulsion along the real axis in \eqref{Gibbs symplectic}.
\end{rem}

Due to the well-known evaluation of $G(\tfrac12)$ (see e.g. \cite{Barnes}) and the recursive relation~\eqref{def of Barnes G}, we have 
\begin{equation} \label{eval of G(3/2)}
\log G(\tfrac32)= \frac{1}{24}\log 2 + \frac32 \zeta'(-1) +\frac14 \log \pi.
\end{equation}
Then as an immediate consequence of the above theorems, we have the following corollary.

\begin{cor} Suppose that $\alpha=c=0$. Then under the same assumptions in Theorem~\ref{Thm_free energy det}, as $N \to \infty$, we have 
\begin{align}
\begin{split} \label{ZN exp det ac0}
\log Z_N^{ \C }(Q) &  = -I[\mu] N^2 +\frac12 N \log N +\Big(  \frac{\log (2\pi)}{2} -1 - 2 I[\mu] -\frac12 E[\mu] \Big) N  + \frac13 \log N
\\
&\quad +\frac{\log(2\pi)}{2}     +2\zeta'(-1)  - I[\mu] - \frac{1}{2} E[\mu] 
 +\frac{ 1  }{6}\log (\rho(0)\tilde{\rho}(0) ) 
\\
&\quad - \frac{11}{12} - \frac{1}{6}\int_{0}^{\infty} \Big(\frac{\rho''(t)}{\rho(t)}-\frac{5}{4}\Big(\frac{\rho'(t)}{\rho(t)}\Big)^2\Big)t\,dt + O(N^{-\frac{1}{12}}(\log N)^{d}),
\end{split}
\end{align}
and 
\begin{align}
\begin{split} \label{ZN exp Pfaff ac0}
\log Z_N^{ \mathbb{H} }(Q) &= -2 I[\mu] N^2 + \frac12 N \log N + \Big( \frac{\log (4\pi)}{2} -1 - 4 I[\mu] -\frac12 E[\mu] - U_\mu(0) \Big) N + \frac{5}{12} \log N 
\\
& \quad  +  \frac{\log (2\pi)}{2} +\frac{5\log 2}{12}  + \zeta'(-1) -2 I[\mu] - \frac{1}{2} E[\mu] -  \, U_\mu(0) + \frac{5}{24} \log (\rho(0)\tilde{\rho}(0) ) 
\\ 
    &\quad    -\frac{17}{24}  
  - \frac{1}{12}\int_0^{\infty} \Big(\frac{\rho''(t)}{\rho(t)} - \frac{5}{4}\Big(\frac{\rho'(t)}{\rho(t)} \Big)^2 \Big)t\,dt + O(N^{-\frac{1}{12}}(\log N)^{d}).  
\end{split}
\end{align}
\end{cor}

\begin{rem}
For a regular and $N$-independent potential $V$, the general conjecture made in \cite{CFTW15,ZW06,JMP94,TF99} (when the associated droplet is connected) reads as
\begin{equation} \label{Z expansion cplx gen}
\begin{split}
\log Z_{N}^\C(V) &= - I_V[\sigma_V] N^2 +\frac12 N \log N + \Big( \frac{\log(2\pi)}{2}-1 - \frac12 \int_\C \log(\Delta V) \,d\sigma_V \Big) N
\\
&\quad + \frac{6-\chi}{12} \log N + \frac{\log(2\pi)}{2} + \chi \, \zeta'(-1) + \FF_V^\C  +o(1),
\end{split}
\end{equation}
where $\chi$ is the Euler index of the droplet and $\FF_V^\C$ is a constant depending on $V$.  
Furthermore, it was formulated in \cite{BKS23} that 
\begin{equation} \label{Z expansion symp gen}
\begin{split}
\log Z_{N}^\H(V) &= - 2\,I_V[\sigma_V] N^2 +\frac12 N \log N + \bigg[ \frac{\log(4\pi)}{2}-1 - \frac12 \int_\C  \log \Big(|z-\bar{z}|^2 \Delta V(z) \Big) \,d\sigma_V(z)   \bigg] N
\\
&\quad + \frac{12-\chi}{24} \log N + \frac{\log(2\pi)}{2} + \frac{\chi}{2} \Big( \frac{5 \log 2}{12}+ \zeta'(-1) \Big)  + \FF_V^\H  +o(1), 
\end{split}
\end{equation}
where $\FF_V^\H$ is a constant depending on $V$, see also \cite[Sections 5.3 and 10.6]{BF24}, \cite[Section 9.3]{Se24} and references therein. 
As previously mentioned, the topology (or $\chi$)-dependent terms appear in the coefficients of $\log N$ and in the constant $O(1)$-terms. Due to the $N$-dependence of our potential, Theorems~\ref{Thm_free energy det} and ~\ref{Thm_free energy Pfaff} do not entirely fit into the formulas above. Nonetheless, in \eqref{ZN exp det ac0} and \eqref{ZN exp Pfaff ac0}, one can observe the conjectural $\chi$-dependent terms with $\chi=2$. 
\end{rem}

\begin{rem}
The growth condition~\eqref{rho:growth} of $\boldsymbol{\rho}$ guarantees that the integrand in $C_5$ and $D_5$ vanishes at infinity with the following decay rate
$$
\frac{\rho''(t)}{\rho(t)}-\frac{5}{4}\Big(\frac{\rho'(t)}{\rho(t)}\Big)^2 = o(\frac{1}{t^2}), \qquad  t\to \infty, 
$$
and thus it is integrable over $[0,\infty)$. 
Let 
\begin{equation}
\kappa(r) = - \frac{2 \Delta \log \boldsymbol{\rho}(z) }{ \boldsymbol{\rho}(z) } =  -\frac{1}{2 \rho(r) } \Big( \frac{\rho'(r)}{ \rho(r) }\Big)', \qquad (r=|z|) 
\end{equation}
be the curvature.
Then 
\begin{equation}
\frac{\rho''(t)}{\rho(t)}-\frac{5}{4}\Big(\frac{\rho'(t)}{\rho(t)}\Big)^2  =  \Big( \frac{\rho'(t)}{ \rho(t) }\Big)' - \frac14 \Big(\frac{\rho'(t)}{\rho(t)}\Big)^2 =-2\kappa(t)\rho(t)  -  \Big( \frac12 \frac{\rho'(t)}{\rho(t)}\Big)^2. 
\end{equation}
This term appears in the Polyakov–Alvarez formula for exterior determinants in \cite[Eq.(4.11)]{ZW06}.
\end{rem}

\begin{rem}
When the equilibrium measure is uniform on the sphere under a particular choice of scaling parameters, the free energy expansion can also be obtained using the Bergman kernel method \cite[Section 3]{Kl14}, see also \cite{FK14,KMMW17}. This method utilizes the fact that the variation of the free energy can be described in terms of the Bergman kernel (in the bulk regime) whose asymptotic expansion is established in differential geometry, see e.g. \cite{Lu00,Tian90,Zel98,MM12,FKZ12} and references therein.
\end{rem}

\begin{rem}
In Theorems~\ref{Thm_free energy det} and ~\ref{Thm_free energy Pfaff}, we focus on smooth potentials without jump discontinuities. On the other hand, when such jump discontinuities exist, the free energy expansions typically become more complicated; for instance, they usually include an additional $O(\sqrt{N})$ term. This situation finds applications in counting statistics \cite{ACCL23,ACCL24,Ch22,ADM24} and hole probabilities \cite{Fo92,Ch23,BP24,Ad18}, where the latter is closely related to the notion of balayage measures in potential theory, see \cite{Ch23a} for recent progress. 
On the other hand, the point charge insertion we consider is also known as a Fisher-Hartwig type singularity, which finds application, for instance, in moments of characteristic polynomials \cite{WW19,DS22,BSY24,BFL25}. In spherical geometry, the point charge insertions, and in particular, their associated equilibrium measure problems, have been extensively studied \cite{CK22,LD21,BFL25}.  
\end{rem}

\begin{ex}
  For the induced spherical ensemble where the potential is given by \eqref{potential spherical ensemble N}, the partition functions can be written explicitly in terms of the Barnes $G$-function as 
  \begin{align*}
 Z_N^{ \mathbb{C} }( Q^{ \rm sp } ) &=     \frac{ N! }{  \Gamma(N+\alpha+c+1)^N  }  \frac{ G(N+c+1) }{ G(c+1) } \frac{ G(N+\alpha+1) }{ G(\alpha+1) },
 \\
  Z_N^{ \mathbb{H} }( Q^{ \rm sp } )  &= N! \Big( \frac{2^{2N+2\alpha+2c+1} }{  \pi \Gamma(2N+2\alpha+2c+2)  } \Big)^N 
 \frac{ G(N+c+1) }{ G(c+1) } \frac{ G(N+c+3/2) }{ G(c+3/2) } \frac{ G(N+\alpha+1) }{ G(\alpha+1) } \frac{ G(N+\alpha+3/2) }{ G(\alpha+3/2) }, 
\end{align*}
see e.g. \cite[Eq.(1.7)]{BP24} or \eqref{ZN random normal symplectic} below. Then their asymptotic behaviours follow from the well-known asymptotic behaviours of the Gamma function \cite[Eq.(5.11.1)]{NIST}
\begin{equation} \label{log N!}
\log N!= N \log N-N+\frac12 \log N+\frac12 \log(2\pi)+O(\frac{1}{N}), \qquad (N \to \infty)
\end{equation}
and also from those of the Barnes $G$-function \cite[Eq.(5.17.5)]{NIST}
\begin{align}
\begin{split} \label{Barnes G asymp}
\log G(z+1) =\frac{z^2 \log z}{2} -\frac34 z^2+\frac{ \log(2\pi) z}{2}-\frac{\log z}{12}+\zeta'(-1)-\frac{1}{240z^2}+O( \frac{1}{z^4} ).
\end{split}
\end{align} 
As a consequence, we have 
\begin{align} \label{free exp det spheri}
\begin{split}
\log  Z_N^{ \mathbb{C} }( Q^{ \rm sp } )  & = -\frac{N^2}{2} + \frac12 N \log N + \Big(\frac{\log (2\pi)}{2} -1  -(\alpha+c) \Big) N +\Big( \frac{\alpha^2+c^2}{2} + \frac{1}{3} \Big)  \log N  
\\
&\quad + \frac{\log(2\pi)}{2} -\frac{1}{12} +2\zeta'(-1) -\frac12 (\alpha+c)(\alpha+c+1-\log(2\pi)) 
\\
&\quad -\log \Big(G(\alpha+1)G(c+1) \Big) +O(\frac{1}{N})
\end{split} 
\end{align} 
and
\begin{align}
\begin{split} \label{free exp Pfaff spheri}
\log  Z_N^{ \mathbb{H} } ( Q^{ \rm sp } )  & = -N^2 + \frac12 N \log N + \Big( \frac{\log (4\pi)}{2}-2-2(\alpha+c) \Big) N +\Big( \alpha^2+\frac{\alpha}{2}+c^2+\frac{c}{2} +\frac5{12}\Big) \log N 
\\
&\quad + \log(2\pi) -\frac{13}{24}+4\zeta'(-1)  -\frac12 (\alpha+c) (2\alpha+2c+3-2\log(2\pi)) 
\\
&\quad -\log \Big( G(\alpha+1) G(\alpha+\tfrac32) G(c+1)G(c+\tfrac32)  \Big) +O(\frac{1}{N}).  
\end{split}
\end{align}
In particular, for the spherical case $\alpha=c=0$, by \eqref{eval of G(3/2)}, we have
\begin{align}
\begin{split}
\log  Z_N^{ \mathbb{C} } ( Q^{ \rm sp } ) \Big|_{\alpha=c=0}  & = -\frac{N^2}{2} + \frac12 N \log N + \Big( \frac{\log (2\pi)}{2} -1 \Big) N  \\
&\quad +\frac1{3} \log N + \frac{\log(2\pi)}{2} -\frac{1}{12} +2\zeta'(-1) +O(\frac{1}{N})
\end{split}
\end{align}
and 
\begin{align}
\begin{split}
 \log  Z_N^{ \mathbb{H} } ( Q^{ \rm sp } ) \Big|_{\alpha=c=0} & = -N^2 + \frac12 N \log N + \Big( \frac{\log (4\pi)}{2}-2 \Big) N 
 \\
 &\quad +\frac5{12} \log N  +\frac{\log(2\pi)}{2} -\frac{13}{24}+ \frac{5\log2}{12} +\zeta'(-1) +O(\frac{1}{N}).  
\end{split}
\end{align}   
The energy and entropy associated with the potential \eqref{potential spherical ensemble infty} are evaluated as   
\begin{equation}
I[\mu]= \frac12, \qquad U_\mu(0)=0, \qquad E[\mu]=  -2. 
\end{equation}
Using these, one can directly observe that Theorems~\ref{Thm_free energy det} and ~\ref{Thm_free energy Pfaff} give rise to \eqref{free exp det spheri} and \eqref{free exp Pfaff spheri}, respectively. 
\end{ex}

\subsection*{Plan of the paper} The rest of this paper is organised as follows. In the next section, we provide preliminaries such as the integrable structures of partition functions and an outline of the proof of our main results. Section~\ref{Section_asymptotic norm} is devoted to the asymptotic behaviours of orthogonal and skew-orthogonal norms. These are crucially used in Section~\ref{Section_proof of theorems}, where we complete the proofs of the main theorems.

\section{Preliminaries and outline of the proof} \label{Section_prelim}

In this section, we provide the preliminaries and outline our main results.

\subsection{Stereographic projection} \label{Subsection_stereographic}

We first provide a proof of Lemma~\ref{Lem_Z sphere plane relation}, which involves a simple change of variables under the stereographic projection.

\begin{proof}[Proof of Lemma~\ref{Lem_Z sphere plane relation}]
Recall that the stereographic projection $\phi$ is given by \eqref{def of stereographic proj}. It has the inverse
\begin{equation}
\phi^{-1}(z)= \Big( \frac{2 \re z}{1+|z|^2}, \frac{2\im z}{1+|z|^2}, \frac{-1+|z|^2}{1+|z|^2} \Big). 
\end{equation}
Let us write $z=\phi(x), w=\phi(y)$ for $x,y \in \mathbb{S}$. Then the chordal distance on the Riemann sphere is given by 
\begin{equation*}
\| x-y \| = \frac{ 2|z-w| }{ \sqrt{ (1+|z|^2)(1+|w|^2) } }.
\end{equation*} 
On the one hand, since 
\begin{equation*}
dA_{ \mathbb S }( x) = \frac{ 
dA(z) }{ (1+|z|^2)^2 },
\end{equation*} 
we have
\begin{align}
\begin{split}
\prod_{j>k=1}^N \| x_j-x_k \|^{2} \prod_{j=1}^{N}  e^{-N Q_{ \mathbb{S} }(x_j) }  \,dA_{ \mathbb S }(x_j)
= 2^{N(N-1)}\prod_{j>k=1}^N |z_j-z_k|^{2} \prod_{j=1}^{N}  e^{-N Q(z_j) }  \,dA(z_j),
\end{split}
\end{align}
and 
\begin{align}
\begin{split}
&\quad \prod_{j>k=1}^N \| x_j-x_k \|^{2} \| x_j-\overline{x}_k \|^{2} \prod_{j=1}^{N} \|x_j-\overline{x}_j\|^2  e^{-2N Q_{ \mathbb{S} }(x_j) }  \,dA_{ \mathbb S }(x_j)
\\
&= 2^{2N^2}\prod_{j>k=1}^N |z_j-z_k|^{2} |z_j-\bar{z}_k|^2 \prod_{j=1}^{N}  |z_j-\bar{z}_j|^2 e^{-N Q(z_j) }  \,dA(z_j).
\end{split}
\end{align}
This completes the proof. 
\end{proof}

\subsection{Partition functions and planar (skew)-orthogonal polynomials} \label{Subsec_ZN and OPs}

We now discuss the integrable structures of the partition functions.
The statistical properties of \eqref{Gibbs cplx} and \eqref{Gibbs symplectic} are effectively analysed using their integrable structures and the associated planar (skew)-orthogonal polynomials, see e.g. \cite{BF24}. 
For a given potential $Q$, let $(p_k)_{ k \in \mathbb{Z} }$ be a family of monic orthogonal polynomial:
\begin{equation} \label{OP norm}
\int_\C p_j(z) \overline{p_k(z)} e^{-N Q(z)}\,dA(z) = h_k \,\delta_{j,k},
\end{equation}
where $h_k$ is the squared norm and $\delta$ is the Kronecker delta. We refer to \cite{HW21,LY23,KKL24} and references therein for recent progress on planar orthogonal polynomials.   
On the other hand, a family of polynomials $(q_k)_{k\in\mathbb{Z}}$ is called planar skew-orthogonal polynomials \cite{Fo13a,Kan02} if it satisfies  
 \begin{equation}
        \langle q_{2k},q_{2\ell} \rangle_s=\langle q_{2k+1},q_{2\ell+1} \rangle_s=0,\qquad 
        \langle q_{2k},q_{2\ell+1} \rangle_s=-\langle q_{2\ell+1},q_{2k} \rangle_s=r_k \, \delta_{k,\ell},
        \label{SQPdef}
    \end{equation}
where 
\begin{equation}
\label{inner_product}
\langle f, g\rangle_s : =\int_{\mathbb{C}} \Big( f(z)\overline{g(z)}-g(z)\overline{f(z)} \Big)(z-\overline{z}) \, e^{-2NQ(z)}\,dA(z)
\end{equation}
and $r_k$ is their skew-norm. 
Then, it follows from Andréief's and de Bruijn’s integration formulas that the partition functions can be expressed as
\begin{equation}
Z_N^\C(Q) = N! \prod_{j=0}^{N-1} h_j, \qquad Z_N^{\mathbb{H}}(Q) = N! \prod_{j=0}^{N-1} r_j.
\end{equation}

For the radially symmetric potential $Q$, it is obvious that the associated orthogonal polynomial is monomials. Therefore, the squared norm is given by $h_{j,N}$, where  
\begin{equation} \label{hj hj tilde def}
h_{j,m}:= \int_\C |z|^{2j}\,e^{-m Q(z)}\,dA(z).
\end{equation}
In general, there is no comprehensive theory for constructing skew-orthogonal polynomials. Nonetheless, a recent work \cite{AEP22} introduced a method to construct them under the assumption that the associated orthogonal polynomial satisfies the three-term recurrence relation, see also a recent work \cite{ABN24} providing an alternative and more general construction. In particular, for a radially symmetric potential $Q$, it follows from \cite[Corollary 4.3]{AEP22} that
\begin{equation} \label{SOP radially symmetric}
q_{2k+1}(z)=z^{2k+1},\qquad q_{2k}(z)=z^{2k}+\sum_{\ell=0}^{k-1} z^{2\ell} \prod_{j=0}^{k-\ell-1} \frac{h_{2\ell+2j+2}}{h_{2\ell+2j+1}}, \qquad r_k=2\,h_{2k+1,2N}. 
\end{equation}
Combining the above, for a radially symmetric $Q$, we have the expression 
\begin{equation} \label{ZN random normal symplectic}
\log Z_N^\C= \log N!+ \sum_{j=0}^{N-1} \log h_{j,N}, \qquad \log Z_N^\H= \log N!+ \sum_{j=0}^{N-1} \log (2h_{2j+1,2N}). 
\end{equation}
The rest of this paper is devoted to analyse the summations in \eqref{ZN random normal symplectic}.

In order to obtain the asymptotic expansions of the orthogonal norms $h_{j,N}$ (resp., $h_{j,2N}$), we distinguish the following three cases. Let $m_1$ and $m_2$ be the large numbers with $m_1$, $m_2 \in \Theta(N^{\epsilon})$ with some $\epsilon>0$.   
\begin{itemize}
    \item Case 1: $m_1 \leq j <N-m_2$ (resp., $2m_1 \leq j < 2(N-m_2)$). Applying the Laplace's method, we obtain the asymptotic expansions of the integral $h_{j,N}$ which mainly contribute the large $N$ expansion of the partition function.
    \smallskip 
    \item Case 2: $0\leq j < m_1$ (resp., $0\leq j < 2m_1$). In this case, the asymptotic expansion contains the gamma function $\Gamma(c+1)$ where $c$ is the point charge at the origin (or the south pole). 
    \smallskip 
    \item Case 3: $N-m_2 \leq j <N$ (resp., $2(N-m_2)\leq j <2N$). Symmetrically to Case 2, the asymptotic expansion contains $\Gamma(\alpha+1)$ where $\alpha$ is the point charge at infinity (or the north pole). 
\end{itemize}
Asymptotic behaviours of the orthogonal norms in each case are given in Lemmas~\ref{lem:logh2}, \ref{lem:logh1}, and \ref{lem:logh3}, respectively.

In addition to the asymptotic behaviours in each case, to analyse the summations
\begin{equation} \label{summations division}
\sum_{j=m_1}^{N-m_2-1} \log h_{j,N}, \qquad \sum_{j=0}^{m_1-1} \log h_{j,N}, \qquad \sum_{j=N-m_2}^{N-1} \log h_{j,N},
\end{equation}
respectively, we will make use of the Euler-Maclaurin formula (see e.g. \cite[Section 2.19]{NIST})
\begin{align}\label{EMfor}
    \sum_{j=m}^n f(j) = \int_m^n f(x)\,dx +\frac{f(m)+f(n)}{2} + \sum_{k=1}^{l-1} \frac{B_{2k}}{(2k)!}\big(f^{(2k-1)}(n) - f^{(2k-1)}(m)\big)+R_l,
\end{align}
where $B_k$ is the Bernoulli number defined by the generating function
$$
\frac{t}{e^t-1} = \sum_{k=0}^\infty B_k \frac{t^k}{k!}, \qquad (|t| < 2\pi). 
$$
Here, the error term $R_l$ satisfies 
\begin{equation*}
    |R_l| \leq C_l \int_m^n |f^{(2l)}(x)|\,dx
\end{equation*}
for some constant $C_l>0$. These are given in Lemmas~\ref{lem:sum2}, ~\ref{lem:sum1} and \ref{lem:sum3}, respectively. 

The asymptotic behaviours in each regime can be interpreted from a geometric viewpoint. In particular, from the perspective of the \(\tau\)-droplet, the droplet associated with the potential $\tau Q$ where \(\tau \in [0,1]\), the index \(j\), ranging from 0 to \(N-1\), corresponds to the progression from the south pole to the north pole. This interpretation provides a natural intuition for the appearance of the gamma function behaviours when \(j < m_1\) or \(j \ge N - m_2\), with dependence on the point charges at the north and south poles, respectively.
An illustration of the asymptotic behaviours in each regime is given in Figure~\ref{Fig_sum division}.

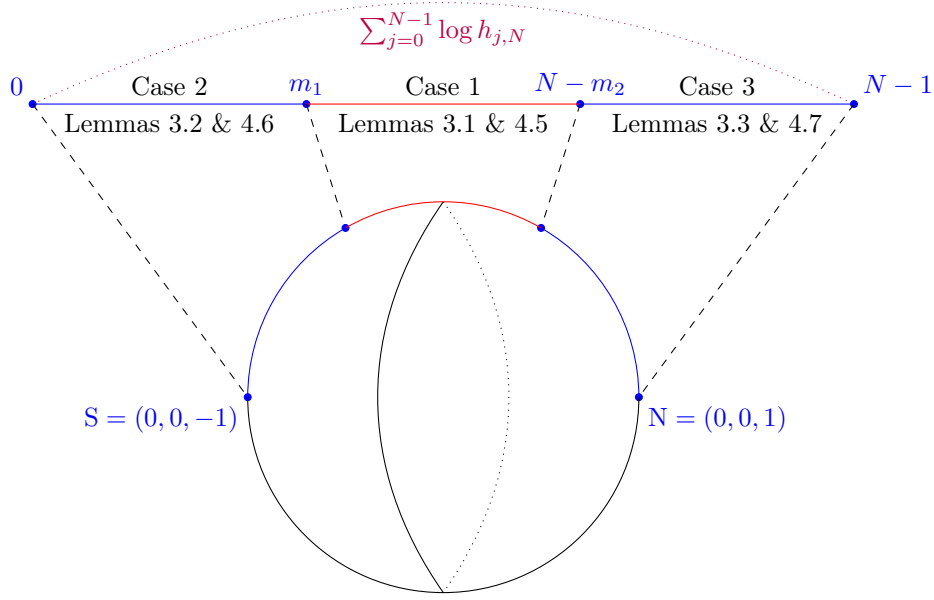
\begin{figure}[t]
\begin{center}
\begin{tikzpicture}[scale = 2.6]
\draw[purple, dotted] (0,0) to [out=25,in=155] node[below] {$\sum_{j=0}^{N-1} \log h_{j,N} $} (4.2,0); 


\filldraw[blue] (0,0) circle (0.5pt) node[above left]{$0$};
\filldraw[blue] (1.4,0) circle (0.5pt) node[above]{$m_1$}; 
\filldraw[blue] (2.8,0) circle (0.5pt) node[above]{$N-m_2$};
\filldraw[blue] (4.2,0) circle (0.5pt) node[above right]{$N-1$};

\draw[blue] (0,0) -- ++(1.4,0)  node[midway,below] {\textcolor{black}{Lemmas~\ref{lem:logh1} \& \ref{lem:sum2}}} node[midway,above] {\textcolor{black}{Case 2}};
\draw[red] (1.4,0) -- ++(1.4,0)  node[midway,below] {\textcolor{black}{Lemmas~\ref{lem:logh2} \& \ref{lem:sum1}}} node[midway,above] {\textcolor{black}{Case 1}};
\draw[blue] (2.8,0) -- ++(1.4,0)  node[midway,below] {\textcolor{black}{Lemmas~\ref{lem:logh3} \& \ref{lem:sum3}}} node[midway,above] {\textcolor{black}{Case 3}};  

\draw[dashed] (0,0) to (1.1,-1.5);  
\draw[dashed] (4.2,0) to (3.1,-1.5);  
\draw[dashed] (1.4,0) to (1.6,-0.633975);  
\draw[dashed] (2.8,0) to (2.6,-0.633975);  

\draw[dotted] (2.1,-0.5) to [out=-55,in=55] (2.1,-2.5);
\draw (2.1,-0.5) to [out=235,in=-235] (2.1,-2.5);

\filldraw[blue] (1.1,-1.5) circle (0.5pt) node[below left]{$\rm S=(0,0,-1)$};
\filldraw[blue] (3.1,-1.5) circle (0.5pt) node[below right]{$\rm N = (0,0,1)$};
\filldraw[blue] (1.6,-0.633975) circle (0.5pt);
\filldraw[blue] (2.6,-0.633975) circle (0.5pt);

\draw (1.1,-1.5) arc[start angle=180, end angle=360, radius=1];
\draw[blue]  (1.6,-0.633975) arc[start angle=120, end angle=180, radius=1];
\draw[blue]  (3.1,-1.5) arc[start angle=0, end angle=60, radius=1];
\draw[red]  (2.6,-0.633975) arc[start angle=60, end angle=120, radius=1];
\end{tikzpicture}
\end{center}
    \caption{Illustration of the different asymptotic regimes of orthogonal norms and their associated lemmas, along with a geometric interpretation of these regimes on spherical geometry.}
    \label{Fig_sum division}
\end{figure}

\subsection{Logarithmic potential and critical points}

In this subsection, we set up some notations that will be used in the asymptotic behaviours of the orthogonal norms.

We first define a function 
\begin{equation}\label{def:F}
    F(t) =  \int_{t}^{\infty} 2r \rho(r)\,dr = 1-\mu(D(0,t)),\qquad t\in [0,\infty).
\end{equation}
It follows from \eqref{rho:growth} that
\begin{equation}
    F(t) \sim \frac{A}{t^2},\qquad t\to \infty. 
\end{equation}
We express the potential $U_\mu$ in \eqref{def of logarithmic potential} in terms of the function $F$ in \eqref{def:F}. 

\begin{lem}\label{lem:U} 
For each $z \in \C \setminus \{0\}$, we have 
\begin{equation} \label{U mu in terms of int F}
    U_{\mu}(z) = -\log|z| - \int_{|z|}^{\infty} \frac{1}{r}F(r) \,dr. 
\end{equation}
\end{lem}
\begin{proof}
Recall the well-known Jensen's formula: for $r>0$, 
\begin{equation} \label{Jensen}
\frac{1}{2\pi} \int_0 ^{2\pi} \log |z-r e^{i\theta}| \,d\theta
=
\begin{cases}
\log r &\text{if }r>|z|,
\smallskip 
\\
\log|z| &\text{otherwise}.
\end{cases}
\end{equation}	
By using this, the definition \eqref{def of logarithmic potential}, and integration by parts, we have 
\begin{align*}
    U_{\mu}(z) &=  \frac{1}{\pi}\int_{0}^\infty \bigg( \int_0^{2\pi} \log \frac{1}{|z-re^{i\theta}|} \, d\theta  \bigg) \rho(r)r\,dr
= -2\log|z| \int_{0}^{|z|}  \rho(r) r\, dr - 2\int_{|z|}^{\infty} r \rho(r) \log r \,dr \\
&= -\big(1-F(|z|)\big)\log|z| + \Big[ F(r)\log r \Big]_{r=|z|}^{\infty}- \int_{|z|}^{\infty} \frac{1}{r}F(r) \,dr.  
\end{align*}   
Then the desired identity follows from the fact that $F(r)\log r \to 0$ as $r\to \infty$. 
\end{proof}

For $r>0$, we define the function  
\begin{equation} \label{def of V_j}
    V_{j}(r) := \frac{N}{\totalN}\Big( Q_{\alpha,c}(r) - \frac{2j+1}{N}\log r\Big)=-2U_{\mu}(r) - \frac{2j+2c+1}{n}\log r, 
\end{equation}
where $Q_{\alpha,c}$ is given by \eqref{def of external potential}. 
It is convenient to use $V_j$ to analyse the orthogonal norm since 
\begin{equation} \label{hj in terms of Vj}
    h_{j,N} = \int_{\C} |z|^{2j} e^{-N Q_{\alpha,c}(z)}\,dA(z) = \int_0^{\infty} 2e^{-\totalN V_j(r)}\,dr.
\end{equation}

For each $\tau$ with $0< \tau < 1$, let $r_\tau$ be the number satisfying 
\begin{equation}\label{def:rtau}
    F(r_\tau)=1-\tau.
\end{equation}
Since $F$ is strictly decreasing in the range between $0$ and $1$, $r_0=0$, $r_1=\infty$, and $r_\tau$ is uniquely determined for each $0<\tau<\infty$.
For each $\alpha, c\ge 0$, we write 
\begin{equation} \label{eq: tau}
   \tau_c(j)= \frac{2(j+c)+1}{2n}, \qquad  \tau_\alpha(j) = \frac{2(j+\alpha)+1}{2n}. 
\end{equation}
Then it follows from the numerology $n=N+\alpha+c+1$ that $\tau_c$ and $\tau_\alpha$ have the relation
\begin{equation}\label{rel:tau}
\tau_c(j) + \tau_\alpha(N-j) = 1.  
\end{equation}

On the other hand, by Lemma~\ref{lem:U}, one can rewrite $V_j$ as 
\begin{equation}\label{def:Vj}
    V_j(r) = 2\Big(\log r + \int_{r}^{\infty} \frac{1}{t}F(t) \,dt - \tau_c(j)\log r \Big).
\end{equation}
This gives 
\begin{equation}\label{dVj}
    V_j'(r) = \frac{2}{r} (1-F(r) - \tau_c(j)).
\end{equation}
Observe that there exists a unique critical point $t_j$ of $V_j$ such that $V_j'(t_j)=0$, and 
\begin{equation} \label{eq: t_j}
    F(t_j)=1-\tau_c(j).
\end{equation}
This implies that 
\begin{equation} \label{def of tj rtau}
t_j = r_{\tau_c(j)}, 
\end{equation} 
where $r_\tau$ is defined in \eqref{def:rtau}.
Differentiating \eqref{dVj}, we have
\begin{align}
\begin{split} \label{differ of Vj 234}
    V_j''(t_j) &= \Big(- \frac{1}{r}  V_j'(r) + 4\rho(r) \Big)\Big|_{r=t_j}=  4\rho(t_j),
    \\
    V_j^{(3)}(t_j) &= \Big(\frac{2}{r^2}V_j'(r) - \frac{4}{r}\rho(r) + 4\rho'(r)\Big)\Big|_{r=t_j} = -\frac{4}{t_j}\rho(t_j) + 4\rho'(t_j),
    \\
     V_j^{(4)}(t_j) &= \Big(-\frac{6}{r^3}V_j'(r) + \frac{12}{r^2}\rho(r) - \frac{4}{r}\rho'(r) + 4\rho''(r)\Big)\Big|_{r=t_j} = \frac{12}{t_j^2}\rho(t_j) - \frac{4}{t_j}\rho'(t_j) + 4\rho''(t_j). 
\end{split}
    \end{align}


The following lemma describes the asymptotic behaviour of the critical point.
\begin{lem}\label{lem:tj}
For small $j\ll N$, the critical point $t_j$ satisfies
\begin{equation}\label{asym:tj}
    t_j = \Big(\frac{\tau_c(j)}{\rho(0)}\Big)^{\frac{1}{2}} + O(\frac{1}{N}),\qquad N\to \infty.
\end{equation}
On the other hand, for $j = N-k$ with small $k\ll N$, it satisfies 
\begin{equation}\label{asym:tk}
    t_j^{-1} = \Big(\frac{\tau_\alpha(k)}{\tilde{\rho}(0)}\Big)^{\frac{1}{2}}  + O(\frac{1}{N}),\qquad N\to \infty.
\end{equation}
\end{lem}
 
\begin{proof}
    For small $t>0$, we have 
    \begin{align*}
        1-F(t) = \int_{0}^{t} 2r\rho(r)\,dr = \rho(0)t^2 + O(t^3), \qquad t\to 0.
    \end{align*}
Combining this with \eqref{eq: t_j}, we have the asymptotic expansion  \eqref{asym:tj} for the critical point $t_j$. 
   On the other hand, for large $t$ near the infinity, by using the inversion, we obtain 
    \begin{align*}
        F(t) = \int_{t}^{\infty} 2r\rho(r)\,dr = \int_{0}^{t^{-1}} 2s^{-3} \rho(s^{-1})\,ds = \tilde{\rho}(0)t^{-2} +O(t^{-3}),\qquad t\to \infty.
    \end{align*}
   Here, we have used $\rho(s) 
    = \tilde{\rho}(0)s^{-4}+O(s^{-5})$ as $s\to\infty$, cf. \eqref{rho:growth}. Combining with the relation \eqref{rel:tau}, this gives the desired result.
\end{proof}

For the symplectic counterpart, we also define 
\begin{equation} \label{eq: symplectic V}
\widetilde{V}_{j}(r):= -2U_{\mu}(r) - \frac{2j+4c+1}{2n}\log r 
\end{equation}
so that 
\begin{equation*}
    h_{j,2N} = \int_{\C} |z|^{2j} e^{-2N Q_{\alpha,c}(z)}\,dA(z) = \int_{0}^{\infty} 2e^{-2n \widetilde{V}_j(r)}\,dr. 
\end{equation*}
The unique critical point $\tilde{t}_j$ of $\widetilde{V}_j$ then satisfies
\begin{equation} \label{eq: symplectic critical point}
F(\tilde{t}_j)=1-\tilde{\tau}_c(j),\qquad \tilde{\tau}_{c}(j) = \frac{2j+4c+1}{4n}. 
\end{equation}
Here, we note that the following relation holds:
\begin{equation}
    \tilde{\tau}_{c}(2(N-j)+1) = \frac{4(N-j)+4c+3}{4n} = 1 - \frac{4j+4\alpha+1}{4n} = 1 - \tilde{\tau}_{\alpha}(2j).
\end{equation}
By applying the same argument as in  Lemma \ref{lem:tj}, we obtain the following asymptotic expansion for $\tilde{t}_j$. 
\begin{lem}\label{lem:tildetj}
For small $j\ll 2N$, the critical point $\tilde{t}_j$ satisfies
\begin{equation}\label{asym:tj tilde}
    \tilde{t}_j = \Big(\frac{\tilde{\tau}_c(j)}{\rho(0)}\Big)^{\frac{1}{2}} + O(\tilde{\tau}_c(j)),\quad N\to \infty.
\end{equation}
On the other hand, for $j = 2N-k$ with small $k\ll 2N$, it satisfies 
\begin{equation}\label{asym:tk tilde}
    \tilde{t}_j^{-1} = \Big(\frac{\tilde{\tau}_\alpha(k+1)}{\tilde{\rho}(0)}\Big)^{\frac{1}{2}} + O(\tilde{\tau}_{\alpha}(k+1)),\quad N\to \infty.
\end{equation}
\end{lem}

\section{Asymptotic analysis of (skew)-orthogonal norms} \label{Section_asymptotic norm}

In this section, we derive the asymptotic behaviours of the (skew)-orthogonal norms for Cases 1, 2, and 3, as described in Subsection~\ref{Subsec_ZN and OPs}.

\subsection{Asymptotics of orthogonal norms: Case 1}

For the formulation of the asymptotic expansion of the orthogonal norm, we define the functional 
\begin{equation}    \mathfrak{B}_1(r):= -\frac{1}{24}\frac{1}{r^2\rho(r)} - \frac{7}{96}\frac{\rho'(r)}{r\rho(r)^2} - \frac{1}{32}\frac{\rho''(r)}{\rho(r)^2} + \frac{5}{96}\frac{(\rho'(r))^2}{\rho(r)^3}.
\end{equation}

\begin{lem} \label{lem:logh2} As $N \to \infty$,we have the following. 
\begin{itemize}
    \item For $j$ with $m_1\leq j < N-m_2$, we have 
\begin{equation}
    \log h_{j,N} = -nV_j(t_j) + \frac{1}{2}\left(\log 2\pi - \log n - \log \rho(t_j)\right) + \frac{1}{n}\mathfrak{B}_1(t_j)+\varepsilon_{N,1}(j),
\end{equation}
where $\varepsilon_{N,1}(j)=O(j^{-2}(\log N)^d) + O((N-j)^{-2}(\log N)^{d})$ for some $d>0$.
\smallskip 
\item  For $j$ with $2m_1 \leq j <  2N-2m_2$, we have 
    \begin{equation}
        \log h_{j,2N} = -2n \widetilde{V}_j(\tilde{t}_j) + \frac{1}{2}\big(\log(2\pi) - \log(2n) - \log \rho(\tilde{t}_j) \big) + \frac{1}{2n}\mathfrak{B}_1(\tilde{t}_j)+\tilde{\varepsilon}_{N,1}(j),
    \end{equation}
    where $\tilde{\varepsilon}_{N,1}(j)=O(j^{-2}(\log N)^d) + O((2N-j)^{-2}(\log N)^{d})$ for some $d>0$.
\end{itemize}
\end{lem}
\begin{proof}
  In the sequel, it is convenient to write 
    \begin{equation} \label{def of delta N}
    \delta_N =  \log N / \sqrt{N} . 
    \end{equation}
Recall that the orthogonal norm $h_{j,N}$ can be written as \eqref{hj in terms of Vj}.  

We first consider the case when $m_1\leq j\leq aN$ for some constant $0<a<1$ so that the critical point $t_j$ given by \eqref{def of tj rtau} is contained in a compact set. Then, by Lemma \ref{lem:tj}, we have  
$$V_j^{(k)}(t_j) = O(t_j^{-k+2}) = O((N/j)^{k/2-1})$$ 
and the Taylor series expansion at the critical point $t_j$ deduces  
\begin{align*}
    &\int_{|r-t_j|<\delta_N} 2\,e^{-\totalN V_j(r)}\,dr\\ 
    =& \frac{2}{\sqrt{\totalN}} e^{-\totalN V_j(t_j)}\int_{-\sqrt{\totalN}\delta_N}^{\sqrt{\totalN}\delta_N} e^{- \big(\frac{1}{2 }V_j''(t_j)s^2 + \frac{1}{6 \sqrt{\totalN}}V_j^{(3)}(t_j)s^3 + \frac{1}{24\totalN} V_j^{(4)}(t_j)s^4 + \frac{1}{120\totalN  \sqrt{\totalN}}V_j^{(5)}(t_j)s^5 + O(j^{-2}|s|^6) \big)}\,ds \\
    =& \frac{2}{\sqrt{\totalN}} e^{-\totalN V_j(t_j)}\int_{-\sqrt{\totalN}\delta_N}^{\sqrt{\totalN}\delta_N} e^{-\frac{1}{2}V_j''(t_j)s^2}\Big(1
    - \frac{1}{24\totalN}V_j^{(4)}(t_j)s^4
    + \frac{1}{72\totalN}(V_j^{(3)}(t_j))^2 s^6   + O(j^{-2}(\log N)^8) \Big)\,ds.
\end{align*}
Computing the Gaussian integrals, we have 
\begin{equation*}
    \begin{split}
        &\int_{|r-t_j|<\delta_N} 2\,e^{-\totalN V_j(r)}\,dr \\
    =& \frac{2\sqrt{2\pi}}{\sqrt{\totalN V_j''(t_j)}} e^{-\totalN V_j(t_j)}\Big( 1 - \frac{1}{8\totalN}\frac{V_j^{(4)}(t_j)}{(V_j''(t_j))^{2}} + \frac{5}{24\totalN}\frac{(V_j^{(3)}(t_j))^2}{(V_j''(t_j))^{3}} + \epsilon_{N,1}(j) \Big)\\
    =& \frac{\sqrt{2\pi}}{\sqrt{\totalN \rho(t_j)}} e^{-\totalN V_j(t_j)}\Big(1 + \frac{1}{\totalN}\Big(-\frac{1}{24}\frac{1}{\rho(t_j)t_j^2} - \frac{7}{96}\frac{\rho'(t_j)}{\rho(t_j)^2 t_j} - \frac{1}{32}\frac{\rho''(t_j)}{\rho(t_j)^2} +\frac{5}{96}\frac{(\rho'(t_j))^2}{\rho(t_j)^3}\Big) + \epsilon_{N,1}(j) \Big),
    \end{split}
\end{equation*}
    where $\epsilon_{N,1}(j)=O(j^{-2}(\log N)^d)$ for some $d>0$ and the $O$-terms are uniform for all $m_1\leq j \leq aN$. 
    
On the other hand, it follows from Lemma \ref{lem:tj} that for $aN\leq j < N-m_2$, 
\begin{equation*}\label{eqn:Vktj}
V_j^{(k)}(t_j) = O(t_j^{-k-2})=O(((N-j)/N)^{k/2+1}).    
\end{equation*}
By using calculations similar to those in the case $m_1\leq j\leq aN$, we have
\begin{equation*}
    \begin{split}
        &\int_{|r-t_j|<\delta_N} 2\,e^{-\totalN V_j(r)}\,dr \\
    =& \frac{\sqrt{2\pi}}{\sqrt{\totalN \rho(t_j)}} e^{-\totalN V_j(t_j)}\Big(1 + \frac{1}{\totalN}\Big(-\frac{1}{24}\frac{1}{\rho(t_j)t_j^2} - \frac{7}{96}\frac{\rho'(t_j)}{\rho(t_j)^2 t_j} - \frac{1}{32}\frac{\rho''(t_j)}{\rho(t_j)^2} +\frac{5}{96}\frac{(\rho'(t_j))^2}{\rho(t_j)^3}\Big) + \epsilon'_{N,1}(j) \Big),
    \end{split}
\end{equation*}
    where $\epsilon'_{N,1}(j)=O((N-j)^{-2}(\log N)^d)$ for some $d>0$ and the $O$-terms are uniform for all $aN\leq j < N-m_2$. 

Next, we show that the integral over the outer region $|r-t_j|>\delta_N$ is negligible. For this, notice that there exists $c>0$ such that for $r$ with $|r-t_j|>\delta_N$, the estimate $$|V_j(r)-V_j(t_j)|\ge c V_j''(t_j)\delta_N^2$$
holds. Then it follows that  
\begin{equation*}
    \begin{split}
        \int_{|r-t_j|>\delta_N} e^{-nV_j(r)}\,dr = e^{-n V_{j}(t_j)}\int_{|r-t_j|>\delta_N} e^{-n(V_j(r)-V_j(t_j))}\,dr \leq  e^{-n V_{j}(t_j)} \epsilon_N, 
    \end{split}
\end{equation*}
where $\epsilon_N=O(e^{-c'(\log N)^2})$ for some $c'>0$ and the $O$-constant can be taken uniformly for all $j$. This completes the proof of the first assertion. The second assertion follows from similar computations with minor modifications.
\end{proof}

\subsection{Asymptotics of orthogonal norms: Case 2 and Case 3}

We begin with the following lemma about the asymptotic expansion of orthogonal norms of Case 2. 

\begin{lem}\label{lem:logh1} As $N \to \infty$ we have the following.
\begin{itemize}
    \item For $j$ with $0\leq j < m_1$, we have 
    \begin{equation}
       \log  h_{j,N} =  2\totalN  U_\mu(0) - (j+c+1)\log
    (\totalN\rho(0) ) + 
    \log\Gamma(j+c+1) + \varepsilon_{N,2}(j),
    \end{equation}
where $\varepsilon_{N,2}(j) = O(N^{-\frac{1}{2}}(j+1)^{\frac{3}{2}}(\log N)^{3})$. 
\smallskip 
\item  For $0\leq j < 2m_1$, we have 
    \begin{equation}
        \log h_{j,2N} = 4n U_{\mu}(0) - (j+2c+1)\log (2n\rho(0)) + \log \Gamma(j+2c+1) + \tilde{\varepsilon}_{N,2}(j),
    \end{equation}
    where $\tilde{\varepsilon}_{N,2}(j)=O(N^{-\frac{1}{2}}(j+1)^{\frac{3}{2}}(\log N)^3)$.
\end{itemize}
\end{lem}

\begin{proof}
Recall that $t_j$ is the critical point of $V_j$ given by \eqref{def of tj rtau}.
Let $t_j^*=t_j\cdot \log N$. We split the integral for $h_{j,N}$ by 
\begin{align}\label{lowerh1}
h_{j,N} = \int_{0}^{t_j^*} 2r^{2j+1}e^{-NQ_{\alpha,c}(r)}  \,dr + \int_{t_j^*}^{\infty} 2r^{2j+1}e^{-NQ_{\alpha,c}(r)} \,dr.
\end{align}
Recall that $Q_{\alpha,c}$ is given by \eqref{def of external potential}. 
Note that by \eqref{U mu in terms of int F}, we have $\d_r U_{\mu}(0)=0$ and $ \d^2_r U_{\mu}(0)= -\rho(0)$. 
This gives 
\begin{align*}
\int_{0}^{t_j^*} 2r^{2j+1}e^{-NQ_{\alpha,c}(r)}\,dr &= \int_{0}^{t_j^*} 2r^{2j+2c+1} e^{2\totalN(U_{\mu}(0) - \frac{1}{2}\rho(0) r^2 +O(r^3) )}\,dr\\
& =e^{2\totalN U_{\mu}(0)}\int_{0}^{t_j^*} 2r^{2(j+c)+1}e^{-\totalN\rho(0)r^2}(1+O(\totalN r^3)) \, dr.
\end{align*}
By using \eqref{def of tj rtau} and \eqref{asym:tj}, we have 
$${\totalN}^{\frac{1}{2}} t_j^* = \Big(\frac{2j+2c+1}{2\rho(0)}\Big)^{\frac{1}{2}}\log N +O((j+1) N^{-\frac{1}{2}}\log N).
$$  
Consequently, it follows that 
\begin{equation}\label{lower:int1}
\begin{split}
    \int_{0}^{t_j^*} 2r^{2j+1}e^{-NQ_{\alpha,c}(r)}\,dr 
    &= e^{2\totalN U_{\mu}(0)} \totalN^{-(j+c+1)}
    \int_{0}^{\infty} 2r^{2j+2c+1}e^{-\rho(0)r^2}\,dr\cdot  (1+O(\totalN\cdot {t_j^*}^3))\\
    &=e^{2nU_{\mu}(0)} (\totalN \rho(0)) ^{-(j+c+1)} \Gamma(j+c+1)\cdot (1+O(N^{-\frac{1}{2}}(j+1)^{\frac{3}{2}}(\log N)^3)).
\end{split}    
\end{equation}

Next, we show that the remaining integral in \eqref{lowerh1} is negligible. 
Notice that by \eqref{Q:growth} and \eqref{def of V_j}, the function $r \mapsto e^{-V_j(r)}$ is integrable.
Since $V_j$ has a global minimum at $t_j$ and increases in $(t_j,\infty)$, for $r>t_j^*$, we have 
\begin{align*}
    V_j(r) \ge V_j(t_j^*) & = V_j(t_j) + 2\rho(t_j)(t_j^*-t_j)^2 + O(t_j^*-t_j)^3 \\
    & = -2U_{\mu}(t_j)-\frac{2(j+c)+1}{\totalN} \log t_j + 2\rho(t_j)(t_j^*-t_j)^2 + O(t_j^*-t_j)^3 \geq -2 U_{\mu}(0) + c_1 {t_j^*}^2, 
\end{align*}
where $c_1 >0$ can be taken to be independent for all $j$ with $j<m_N$.
Then it follows that 
\begin{equation}\label{lower:int2}
\begin{split}
 \int_{t_j^*}^{\infty} 2e^{-\totalN V_{j}(r)}\,dr \leq e^{-(\totalN-1)(-2U_{\mu}(0)+c_1 {t_j^*}^2)} \int_{t_j^*}^{\infty} e^{-V_j(r)}\,dr \leq C e^{2\totalN U_{\mu}(0)} e^{-c_1(\totalN-1){t_j}^2(\log N)^2} 
\end{split}
\end{equation}
for some $C>0$. Therefore we conclude that the integral in \eqref{lower:int2} is negligible. 
\end{proof}

We use the inversion to obtain the estimate of the orthogonal norm $h_{j,N}$ for $N-m_2 \leq j\leq N-1$.

\begin{lem}\label{lem:logh3} As $N \to \infty$, we have the following.
\begin{itemize}
    \item    For $j$ with $N-m_2\leq j < N$, we have 
    \begin{equation}
        \log h_{j,N} = -(N-j+\alpha)\log(\totalN \tilde{\rho}(0)) +\log\Gamma(N-j+\alpha)+ \varepsilon_{N,3}(j), 
    \end{equation}
    where $\varepsilon_{N,3}(j) = O(N^{-\frac{1}{2}}(N-j)^{\frac{3}{2}}(\log N)^3)$.
    \smallskip 
    \item  For $2N - 2m_2 \leq j < 2N $, we have  
    \begin{equation}
        \log h_{j,2N} = - (2N-j+2\alpha +1) \log (2n\tilde{\rho}(0)) + \log \Gamma(2N-j+2\alpha+1)+ \tilde{\varepsilon}_{N,3}(j),
    \end{equation}
    where $\tilde{\varepsilon}_{N,3}(j)=O(N^{-\frac{1}{2}}(2N-j)^{\frac{3}{2}}(\log N)^3).$
\end{itemize}
\end{lem}

\begin{proof}
 Using the inversion, we write
    \begin{equation*}
        h_{j,N} = \int_{0}^{\infty} 2e^{-\totalN W_j(s)} \,ds, \qquad   W_j(s) := V_j(s^{-1})+\frac{2}{\totalN}\log s. 
    \end{equation*} 
Note that by \eqref{rel:tau} and \eqref{def:Vj}, we have
\begin{equation*}
\begin{split}
    W_j(s) 
    = 2\int_{0}^{s} u^{-1}F(u^{-1}
    )\,du - 2\tau_{\alpha}(N-1-j)\log s, \qquad 
    W_j'(s) = 2s^{-1} \Big(  F(s^{-1}) - \tau_\alpha(N-1-j) \Big), 
\end{split}
\end{equation*} 
where $F$ is given by \eqref{def:F}. 
Notice that there exists a unique number $s_j$ such that $W_j'(s_j)=0$. 
Then, similar to the proof of Lemma~\ref{lem:tj}, we can see that
for $j$ with $N-m_2 \leq j \leq N-1,$
\begin{equation}\label{asymsj}
    s_j = \Big(\frac{\tau_{\alpha}(N-1-j)}{\tilde{\rho}(0)}\Big)^{\frac{1}{2}} + O(\tau_{\alpha}(N-1-j)),\quad N\to \infty,
\end{equation}
where $\tilde{\rho}$ is defined by \eqref{def of rho tilde}. Note that by \eqref{rho:growth}, $\tilde{\rho}$ has the asymptotic expansion 
\begin{equation}\label{asym:trho}
    \tilde{\rho}(s) = \tilde{\rho}(0) + \tilde{\rho}'(0)s + \frac{1}{2}\tilde{\rho}''(0)s^2 + O(s^3),\qquad s\to 0.
\end{equation} 
By using the change of variables, we have
\begin{equation*}
\begin{split}
    \int_{0}^{s} u^{-1}F(u^{-1})\,du & = \int_{0}^{s} u^{-1}\int_{u^{-1}}^{\infty} 2r\rho(r)\,dr \,du  = \int_{0}^s u^{-1} \int_0^{u} 2w\tilde{\rho}(w) \,dw \,du = \frac{1}{2}\tilde{\rho}(0)s^2 + O(s^3),
\end{split}
\end{equation*}
as $s \to 0$.
Let $s_j^* = s_j \log N$. Then it follows from \eqref{asymsj} that 
$$
\totalN^{\frac{1}{2}}s_j^* = \Big( \frac{2(N-j+\alpha)-1}{2\tilde{\rho}(0)} \Big)^{\frac{1}{2}}\log N + O(N^{-\frac{1}{2}}(N-j)\log N) .
$$
By using this asymptotic behaviour, we obtain 
\begin{align*}
   \int_0^{s_j^*} 2e^{-\totalN W_j(s)}\,ds &= \int_{0}^{s_j^*} 2s^{2(N-j+\alpha)-1}e^{-\totalN \tilde{\rho}(0)s^2}(1+O(\totalN s^3))\,ds
   \\
   & = (\totalN \tilde{\rho}(0))^{-(N-j+\alpha)}\,\Gamma(N-j+\alpha)\cdot (1+O(N^{-\frac{1}{2}}(N-j)^{\frac{3}{2}}(\log N)^3)). 
\end{align*}
Furthermore, as in the proof of Lemma \ref{lem:logh1}, the integral over $(s_j^*, \infty)$ is negligible. Hence, the desired result follows. 
\end{proof}

\section{Proof of main results} \label{Section_proof of theorems}

In this section, we prove our main results, Theorems~\ref{Thm_free energy det} and ~\ref{Thm_free energy Pfaff}. 

\subsection{Proof of Theorem~\ref{Thm_free energy det}}

In this subsection, we show Theorem~\ref{Thm_free energy det}.
For this, we need to derive the asymptotic expansion of the summations \eqref{summations division} using Lemmas~\ref{lem:logh2}, ~\ref{lem:logh1} and ~\ref{lem:logh3}. 

We begin with the summation $\sum_{j=m_1}^{N-m_2-1} \log h_{j,N}$. 
Note that by Lemma~\ref{lem:logh2}, it suffices to derive asymptotic behaviours of the summations 
\begin{equation} \label{summations in logh2}
-n \sum_{j=m_1}^{N-m_2-1} V_j(t_j) , \qquad \frac12 \sum_{j=m_1}^{N-m_2-1} \log \rho(t_j) , \qquad    \frac{1}{n}  \sum_{j=m_1}^{N-m_2-1} \mathfrak{B}_1(t_j). 
\end{equation}
The asymptotic behaviour of the first term is given by \eqref{eq:EMVj} together with Lemmas~\ref{lem:vsum1} and ~\ref{lem:vsum2}, whereas the asymptotic behaviours of the second and third terms are given by Lemmas~\ref{lem:vsum3} and ~\ref{lem:vsum4}, respectively. Combining all of the above, the final result for the asymptotic behaviour of $\sum_{j=m_1}^{N-m_2-1} \log h_{j,N}$ is provided in Lemma~\ref{lem:sum1}. 

Recall that $V_j$ is given by \eqref{def of V_j}. We observe that 
$$
V_j(t_j) = -2 U_{\mu}(t_j) - 2\tau_c(j)\log t_j = -2U_{\mu}(t_j) - 2(1-F(t_j))\log t_j,
$$
where $t_j$ is given by \eqref{def of tj rtau}.
Note that the definition of $t_j$ can be extended from integer values $j$ to positive real numbers $s$. Since $F'(t)=-2t\rho(t)$, we also have 
\begin{equation} \label{eq: dt_s}
    \frac{dt_s}{ds} = -\frac{1}{nF'(t_s)} = \frac{1}{2\totalN  t_s\rho(t_s)}.
\end{equation}
Then by the Euler-Maclaurin formula \eqref{EMfor}, we have 
\begin{equation}\label{eq:EMVj}
    \begin{split}
        \sum_{j=m_1}^{N-m_2-1} V_j(t_j) 
        &= -\totalN \int_{t_{m_1}}^{t_{N-m_2}}\Big( 2U_{\mu}(t) + 2(1-F(t))\log t \Big)\cdot 2t\rho(t)\,dt
        \\
        &\quad + \frac{1}{2}\Big(V_{m_1}(t_{m_1})-V_{N-m_2}(t_{m_2}) \Big) + \frac{1}{12}\Big( \d_s ( V_s(t_s))\big|_{s={N-m_2}} - \d_s ( V_s(t_s))\big|_{s={m_1}} \Big) + \epsilon_N.
    \end{split}
\end{equation}
Here, the error term satisfies 
$$
\epsilon_N = O(N^{-3}\tau_c(m_1)^{-2}) + O(N^{-3}\tau_\alpha(m_2)^{-2})=O(N^{-1}m_1^{-2})+O(N^{-1}m_2^{-2})
$$ 
since $\d_s^3\,V_s(t_s)\big|_{s=m_1} = O(N^{-3}t_{m_1}^{-4})$ and $\d_s^3\, V_s(t_s)|_{s=N-m_2} = O(N^{-3}t_{N-m_2}^4)$, cf. \eqref{differ of Vj 234}.


We first derive the asymptotic behaviour of the first term in the left-hand side of \eqref{eq:EMVj}. 

\begin{lem}\label{lem:vsum1}
As $N\to \infty$, we have
\begin{align}
\begin{split}
        &\quad -n\int_{t_{m_1}}^{t_{N-m_2}} \Big(2U_{\mu}(t) + 2(1-F(t))\log t \Big)\cdot 2t\rho(t)\,dt 
        \\
        &= nI[\mu] + n\Big(2\tau_c(m_1)U_{\mu}(0) - \frac{3}{4}\tau_c(m_1)^{2} + \tau_c(m_1)^2\log t_{m_1}\Big)
        -n\Big(\frac{3}{4} \tau_{\alpha}(m_2)^2 +\tau_{\alpha}(m_2)^2 \log t_{N-m_2}\Big) + \epsilon_N,
\end{split}
\end{align} 
where $\epsilon_N = O\big(N^{-\frac{3}{2}}(m_1^{\frac{5}{2}}+m_2^{\frac{5}{2}})\big)$.
\end{lem}

\begin{proof}
By using $I[\mu] = - \int U_{\mu} \,d\mu$, we first write 
\begin{equation*}
    \begin{split}
        -\int_{t_{m_1}}^{t_{N-m_2}} U_{\mu} (t)\, 2t\rho(t)\,dt = I[\mu] + \int_{0}^{t_{m_1}} U_{\mu}(t)\,2t\rho(t)\,dt + \int_{t_{N-m_2}}^{\infty} U_{\mu}(t)\,2t\rho(t)\,dt.
    \end{split}
\end{equation*}
Note that by Lemma \ref{lem:U} and \eqref{def:F}, we have  
\begin{equation*}
    U_{\mu}'(t) = - \frac{1}{t}(1-F(t)),  \qquad  2t\rho(t) = (1-F(t))'.
\end{equation*}
Here, by the Taylor series expansion at $t=0$ together with \eqref{def of tj rtau}, we have 
\begin{equation*}
    \begin{split}
        &\quad \int_0^{t_{m_1}}U_{\mu}(t)\cdot  2t\rho(t)\,dt 
        = \Big[(1-F(t))U_{\mu}(t)\Big]_{t=0}^{t_{m_1}} - \int_0^{t_{m_1}} (1-F(t))U_{\mu}'(t)\,dt 
        \\
        & = \tau_c(m_1)U_{\mu}(t_{m_1}) + \int_0^{t_{m_1}} t \, U_{\mu}'(t) U_{\mu}'(t)\,dt = \tau_c(m_1)U_{\mu}(t_{m_1}) + \int_0^{t_{m_1}} t\Big(U_{\mu}''(0)t + O(t^2)\Big)^2\,dt 
        \\
        & = \tau_c(m_1)U_{\mu}(t_{m_1}) + U_{\mu}''(0)^2 \frac{1}{4}t_{m_1}^4 + O(t_{m_1}^5). 
    \end{split}
\end{equation*}
Similarly, we have
\begin{equation*}
    \begin{split}
     \int_{t_{N-m_2}}^{\infty}U_{\mu}(t)\cdot 2t\rho(t)\,dt 
        &= -\Big[U_{\mu}(t)F(t)\Big]_{t=t_{N-m_2}}^{\infty} 
        \\
        &\quad + \int_{t_{N-m_2}}^{\infty} U_{\mu}'(t)F(t)\,dt =  \tau_\alpha(m_2)U_{\mu}(t_{N-m_2}) - \int_{t_{N-m_2}}^{\infty} \frac{1}{t}(1-F(t))F(t)\,dt.
    \end{split}
\end{equation*}
Note that by \eqref{rho:growth}, we have 
\begin{equation*}
\begin{split}
    F(s^{-1}) 
    &= \tilde{\rho}(0)s^2 + \frac{2}{3}\tilde{\rho}'(0)s^3 + \frac{1}{4}\tilde{\rho}''(0)s^{4} + O(s^5),\\ 
    U_{\mu}(s^{-1})& = \log s - \frac{1}{2}\tilde{\rho}(0)s^2 - \frac{2}{9}\tilde{\rho}(0)s^3 - \frac{1}{16}\tilde{\rho}''(0)s^4 + O(s^5),
\end{split}
    \end{equation*}
as $s\to 0$.
This gives rise to 
\begin{equation*}
    \begin{split}
        \int_{t_{N-m_2}}^{\infty} \frac{1}{t}(1-F(t))F(t)\,dt 
        &= \int_{0}^{t_{N-m_2}^{-1}} \frac{1}{s}(1-F(s^{-1}))F(s^{-1})\,ds \\
        &= \frac{1}{2}\tilde{\rho}(0) t_{N-{m_2}}^{-2} + \frac{2}{9}\tilde{\rho}'(0)t_{N-{m_2}}^{-3}-\Big(\tilde{\rho}(0)^2-\frac{1}{4}\tilde{\rho}''(0)\Big)\frac{1}{4}t_{N-{m_2}}^{-4}+O(t_{N-{m_2}}^{-5}).
    \end{split}
\end{equation*}
We also have
\begin{equation*}
    \begin{split}
        \int_{t_{m_1}}^{t_{N-m_2}} (U_{\mu}(t) + 2(1-F(t))\log t)\cdot 2t\rho(t)\,dt = \Big[(1-F(t))^2\log t + (1-F(t))U_{\mu}(t)\Big]_{t=t_{m_1}}^{t_{N-m_2}}.
    \end{split}
\end{equation*}
Combining all of the asymptotic expansions, we have
\begin{equation*}
    \begin{split}
        -n\int_{t_{m_1}}^{t_{N-m_2}}(2U_{\mu}(t)+2(1-F(t))\log t) \cdot 2t\rho(t)\,dt 
        = nI[\mu] + n\mathcal{E}_{1,1}(m_1)+n\mathcal{E}_{1,2}(m_2) + O(nt_{m_1}^{5}+nt_{N-m_2}^{-5}),
    \end{split}
\end{equation*}
where 
\begin{equation*}
    \begin{split}
        \mathcal{E}_{1,1}(m_1) &= 2\tau_{c}(m_1)U_{\mu}(t_{m_1}) + U_{\mu}''(0)^2\frac{1}{4}t_{m_1}^4+ \tau_c(m_1)^2\log t_{m_1}, \\
        \mathcal{E}_{1,2}(m_2) &=  2\tau_\alpha (m_2)U_{\mu}(t_{N-m_2}) - U_{\mu}(t_{N-m_2}) - \tau_c(N-m_2)^2\log t_{N-m_2} \\ 
        &\quad -\frac{1}{2}\tilde{\rho}(0) t_{N-m_2}^{-2} - \frac{2}{9}\tilde{\rho}'(0)t_{N-m_2}^{-3} + \Big(\tilde{\rho}(0)^2 - \frac{1}{4}\tilde{\rho}''(0)\Big)\frac{1}{4}t_{N-m_2}^{-4}.
    \end{split}
\end{equation*}
Furthermore, by straightforward computations, we obtain 
\begin{equation*}
    \begin{split}
        \mathcal{E}_{1,1}(m_1)
        &=2\tau_c(m_1)U_{\mu}(0) - \frac{3}{4}\tau_c(m_1)^{2} + \tau_c(m_1)^2\log t_{m_1} + O(\tau_c(m_1)^{\frac{5}{2}}),
        \\
        \mathcal{E}_{1,2}(m_2) &=-\tau_{\alpha}(m_2)^2 \log t_{N-m_2}  -\frac{3}{4} \tau_{\alpha}(m_2)^2 + O(\tau_\alpha(m_2)^{\frac{5}{2}}),
    \end{split}
\end{equation*}
which completes the proof.
\end{proof}

Next, we turn to asymptotic behaviours the remaining terms in \eqref{eq:EMVj}. 
\begin{lem}\label{lem:vsum2}
    As $N\to \infty$, we have 
    \begin{align}
    \begin{split}
    &\quad \frac{1}{2}\Big(V_{m_1}(t_{m_1})- V_{N-m_2}(t_{N-m_2})\Big)
    \\ 
    &= - U_{\mu}(0) + \frac{1}{2}\tau_c(m_1)     - \tau_c(m_1)\log t_{m_1}  -\frac{1}{2}\tau_{\alpha}(m_2)    - \tau_{\alpha}(m_2)\log t_{N-m_2}+O\big(N^{-\frac{3}{2}}(m_1^{\frac{3}{2}}+m_2^{\frac{3}{2}})\big)  
    \end{split}
    \end{align}
    and 
    \begin{equation}
        \begin{split}
    \frac{1}{12}\Big( \d_s (V_s(t_s))\big|_{s=N-m_2} - \d_s (V_s(t_s))\big|_{s=m_1} \Big)     =\frac{1}{6n}\Big(\log t_{m_1} - \log t_{N-m_2}\Big).
        \end{split}
    \end{equation}
\end{lem}

\begin{proof}
The first assertion follows from 
\begin{equation*}
    \begin{split}
        &\quad \frac{1}{2}(V_{m_1}(t_{m_1})- V_{N-m_2}(t_{N-m_2}))
        \\ 
        &= -U_{\mu}(t_{m_1})-\tau_c(m_1)\log t_{m_1} + U_{\mu}(t_{N-m_2}) + \tau_c(N-m_2)\log t_{N-m_2}
        \\
        &= - U_{\mu}(0) + \frac{1}{2}\tau_c(m_1) - \tau_c(m_1)\log t_{m_1} - \tau_{\alpha}(m_2)\log t_{N-m_2}-\frac{1}{2}\tau_{\alpha}(m_2)+O(\tau_c(m_1)^{\frac{3}{2}}+\tau_{\alpha}(m_2)^{\frac{3}{2}}),
    \end{split}
\end{equation*}
where we use the relation \eqref{rel:tau} between $\tau_c$ and $\tau_\alpha$.
Applying the chain rule, 
\begin{equation*}
    \d_s(V_s(t_s)) = \d_s(-2U_{\mu}(t_s) - 2\tau_c(s)\log t_s) = -V'_s(t_s) \frac{dt_s}{ds} - \frac{2}{n}\log t_s = -\frac{2}{n}\log t_s,
\end{equation*}
which yields the second assertion. 
\end{proof}

Next, we derive the asymptotic behaviour of the second summation in \eqref{summations in logh2}. 

\begin{lem}\label{lem:vsum3}
As $N \to \infty$, we have
\begin{align}
\begin{split}
  \sum_{j=m_1}^{N-m_2-1}\frac{1}{2}\log \rho(t_j) &= \frac{1}{2}\totalN E[\mu] - \frac{1}{2}\totalN \tau_c(m_1)\log \rho(0) - \frac{1}{2}\totalN \tau_\alpha(m_2)(-2+\log \tilde{\rho}(0) - 4\log t_{N-m_2})
        \\
        &\quad +\frac{1}{4}(\log \rho(t_{m_1})-\log \rho(t_{N-m_2})) + O(m_1^{-1}+m_2^{-1})+O(N^{-\frac{1}{2}}(m_1^{\frac{3}{2}}+m_2^{\frac{3}{2}})\log N). 
\end{split}
\end{align} 
\end{lem}
\begin{proof}
 By using the Euler-Maclaurin formula \eqref{EMfor} and \eqref{eq: dt_s}, we have
    \begin{equation*}
        \begin{split}
            \sum_{j=m_1}^{N-m_2-1} \frac{1}{2}\log \rho(t_j) = n\int_{t_{m_1}}^{t_{N-m_2}}\frac{1}{2}\log \rho(t) \cdot 2t\rho(t)\,dt + \frac{1}{4}(\log \rho(t_{m_1})-\log\rho(t_{N-m_2}))+ O(N^{-1}(t_{m_1}^{-2}+t_{N-m_2}^{2})).
        \end{split}
    \end{equation*}
Using the Taylor expansion at $0$ and \eqref{asym:tj}, we have
\begin{align*}
  \int_{0}^{t_{m_1}} \log \rho(t) \cdot 2t\rho(t)\,dt 
        &= \int_{0}^{t_{m_1}} \left(\rho(0)\log \rho(0) + O(t)\right)2t\,dt \\
        &= \rho(0) \log\rho(0) t_{m_1}^2 + O(t_{m_1}^3) =  \log\rho(0)\tau_c(m_1) + O(t_{m_1}^3).
\end{align*} 
On the other hand, it follows from \eqref{asym:trho} and \eqref{asym:tk} that 
\begin{align*}
   \int_{t_{N-m_2}}^{\infty}\log \rho(t) \cdot 2t\rho(t)\,dt 
        &= \int_{0}^{t_{N-m_2}^{-1}} \log (s^4\tilde{\rho}(s)) \cdot 2s\tilde{\rho}(s)\,ds \\  
        &=\int_{0}^{t_{N-m_2}^{-1}} \left( 4\log s + \log \tilde{\rho}(0) + O(s)\right)\cdot (2\tilde{\rho}(0)s + O(s^2))\,ds \\
        &= \tilde{\rho(0)}t_{N-m_2}^{-2}(-2+\log \tilde{\rho}(0) - 4\log t_{N-m_2})+O(t_{N-m_2}^{-3}\log t_{N-m_2}) \\
        &= \tau_\alpha(m_2)(-2 +\log \tilde{\rho}(0) - 4\log t_{N-m_2})+O(t_{N-m_2}^{-3}\log t_{N-m_2}).
\end{align*} 
Since the error terms satisfy  $O(N^{-1}(t_{m_1}^{-2}+t_{N-m_2}^{2})) = O(N^{-1}(\tau_c(m_1)^{-1}+\tau_\alpha(m_2)^{-1}))$, 
$O(t_{m_1}^3) = O(\tau_c(m_1)^{\frac{3}{2}})$ and $O(t_{N-m_2}^{-3}\log t_{N-m_2})=O(\tau_{\alpha}(m_2)^{\frac{3}{2}}\log N)$, the proof is complete. 
\end{proof}

Finally, we derive the asymptotic behaviour of the third summation in \eqref{summations in logh2}.

\begin{lem}\label{lem:vsum4}
  As $N \to \infty$, we have 
  \begin{align}
\begin{split}
   \sum_{j=m_1}^{N-m_2-1}\frac{1}{n}\mathfrak{B}_1(t_j)& = -\frac{1}{12}\log\Big(\frac{t_{N-m_2}}{t_{m_1}}\Big)-\frac{1}{4}\log\Big(\frac{\rho(t_{N-m_2})}{\rho(t_{m_1})}\Big)-\frac{5}{12}-\frac{1}{6}\int_{0}^{\infty}\Big(\frac{\rho''(t)}{\rho(t)}-\frac{5}{4}\Big(\frac{\rho'(t)}{\rho(t)}\Big)^2\Big)t\,dt \\
            &\quad +O(m_1^{-1}+m_2^{-1}).
\end{split}
  \end{align} 
\end{lem}

\begin{proof}
By using the Euler-Maclaurin formula \eqref{EMfor} and \eqref{eq: dt_s}, 
\begin{equation}\label{eq:intB1}
    \begin{split}
        &\quad \sum_{j=m_1}^{N-m_2-1}\frac{1}{n}\mathfrak{B}_1(t_j) 
        = \int_{t_{m_1}}^{t_{N-m_2}}\mathfrak{B}_1(t) 2t\rho(t)\,dt + O(n^{-1}(t_{m_1}^{-2}+t_{N-m_2}^2))\\
        &=\int_{t_{m_1}}^{t_{N-m_2}}
        \Big(-\frac{1}{24}\frac{1}{t^2\rho(t)} - \frac{7}{96}\frac{\rho'(t)}{t\rho(t)^2}-\frac{1}{32}\frac{\rho''(t)}{\rho(t)^2}+\frac{5}{96}\frac{(\rho'(t))^2}{\rho(t)^3} \Big)\cdot 2t\rho(t)\,dt + O(m_1^{-1}+m_2^{-1})\\
        &=\int_{t_{m_1}}^{t_{N-m_2}}
        \Big(-\frac{1}{12t} - \frac{7}{48}\frac{\rho'(t)}{\rho(t)}-\frac{1}{16}\frac{\rho''(t)}{\rho(t)}t + \frac{5}{48}\Big(\frac{\rho'(t)}{\rho(t)}\Big)^2 t \Big)\,dt + O(m_1^{-1}+m_2^{-1}).
    \end{split}
\end{equation}
Using integration by parts, we have 
\begin{equation*}
    \begin{split}
        \int_{t_{m_1}}^{t_{N-m_2}} \Big(\frac{\rho''(t)}{\rho(t)}-\Big(\frac{\rho'(t)}{\rho(t)}\Big)^2 \Big)t \,dt 
        & = \left[\frac{\rho'(t)}{\rho(t)}\,t\right]_{t=t_{m_1}}^{t_{N-m_2}} - \int_{t_{m_1}}^{t_{N-m_2}} \frac{\rho'(t)}{\rho(t)}\,dt.
    \end{split}
\end{equation*}
To obtain integrable terms for the integral in \eqref{eq:intB1}, we again use the above integration by parts formula and obtain
\begin{equation*}
    \begin{split}
        \sum_{j=m_1}^{N-m_2-1}\frac{1}{n}\mathfrak{B}_1(t_j) &= -\int_{t_{m_1}}^{t_{N-m_2}} \frac{1}{12t} + \frac{1}{4}\frac{\rho'(t)}{\rho(t)} + \frac{1}{6}\Big(\frac{\rho''(t)}{\rho(t)}-\frac{5}{4}\frac{\rho'(t)^2}{\rho(t)^2}\Big)t \,dt 
    + \frac{5}{48}\left[\frac{\rho'(t)}{\rho(t)}t\right]_{t=t_{m_1}}^{t_{N-m_2}}\! + O(m_1^{-1}+m_2^{-1}).
    \end{split}
\end{equation*}
Furthermore, it follows from  \eqref{rho:growth} that $$\frac{\rho'(t_{N-m_2})}{\rho(t_{N-m_2})}t_{N-m_2}=-4+O(t_{N-m_2}^{-1}),\qquad N\to \infty,$$
which completes the proof.
\end{proof}

We now derive the desired asymptotic behaviour of $\sum_{j=m_1}^{N-m_2-1}\log h_{j,N}$. 

\begin{lem}\label{lem:sum1}
As $N \to \infty$, we have 
\begin{align}
\begin{split}
\sum_{j=m_1}^{N-m_2-1}\log h_{j,N} &= -n^2 I[\mu] + \frac{1}{2}(N-m_1-m_2)\log\Big(\frac{2\pi}{n}\Big) - \frac{1}{2}nE[\mu] - 2n({m_1}+c)U_{\mu}(0) 
\\
&\quad -\Big((m_1+c)^2-\frac{1}{6}\Big)\log t_{m_1} +\Big((m_2+\alpha)^2-\frac{1}{6}\Big)\log t_{N-m_2}
\\ 
&\quad +\frac{3}{4}n^2(\tau_c(m_1)^2+\tau_\alpha(m_2)^2) - \frac{n}{2}\tau_c(m_1) (1-\log\rho(0))-\frac{n}{2}\tau_\alpha(m_2)(1-\log \tilde{\rho}(0)) \\
        & \quad  -\frac{5}{12} - \frac{1}{6}\int_{0}^{\infty} \Big(\frac{\rho''(t)}{\rho(t)}-\frac{5}{4}\Big(\frac{\rho'(t)}{\rho(t)}\Big)^2\Big)t\,dt
        \\
        &\quad       +O((m_1^{-1}+m_2^{-1})(\log N)^d)+O(N^{-\frac{1}{2}}(m_1^{\frac{5}{2}}+m_2^{\frac{5}{2}})).
\end{split}
\end{align} 
\end{lem}

\begin{proof}
Recall that the error term $\varepsilon_{N,1}(j)$ in Lemma \ref{lem:logh2} has a bound $O(j^{-2}(\log N)^{d})+O((N-j)^{-2}(\log N)^{d})$ for some constant $d$ and the $O$-constant can be taken uniformly for all $j$ with $m_1 \leq j \leq N-m_2$. Therefore we have 
\begin{equation*}
    \sum_{j=m_1}^{N-m_2-1} \varepsilon_{N,1}(j) = O((m_1^{-1}+m_2^{-1})(\log N)^d).
\end{equation*}
On the other hand, by combining \eqref{eq:EMVj} with Lemmas \ref{lem:vsum1} and \ref{lem:vsum2}, we have 
\begin{equation*}
    \begin{split}
        -\sum_{j=m_1}^{N-m_2-1} \totalN V_j(t_j) 
        &= -\totalN^2 I[\mu] 
        -2\totalN^2\tau_c(m_1)U_{\mu}(0) + \totalN^2\tau_c(m_1)^2\Big( \frac{3}{4}-\log t_{m_1}\Big) + \totalN^2\tau_{\alpha}(m_2)^2\Big(\frac{3}{4}+\log t_{N-m_2}\Big)
        \\
        &\quad  +  nU_{\mu}(0) - \totalN \tau_c(m_1)\Big( \frac{1}{2} - \log t_{m_1}\Big) + \totalN \tau_{\alpha}(m_2)\Big( \frac{1}{2} + \log t_{N-m_2}\Big)
        \\
        & \quad - \frac{1}{6}\Big( \log t_{m_1} - \log t_{N-m_2}\Big) 
        + O(N^{-\frac{1}{2}}(m_1^{\frac{5}{2}}+m_2^{\frac{5}{2}})) + O(m_1^{-2}+m_2^{-2}).
    \end{split}
\end{equation*}
Then straightforward computations using Lemma \ref{lem:logh2}, together with this asymptotic behaviour and those in Lemmas \ref{lem:vsum3} and \ref{lem:vsum4} complete the proof. 
\end{proof}

Recall that we have treated three different regimes (see Figure~\ref{Fig_sum division} and the text above) and in Lemma~\ref{lem:sum1}, we derive the asymptotic behaviour for the Case 1. 
In the following two lemmas, we show the counterparts for the Case 2 and Case 3, respectively.

\begin{lem}\label{lem:sum2}
As $N \to \infty$, we have 
\begin{align}
\begin{split}
  \sum_{j=0}^{m_1-1} \log h_{j,N} &= 2m_1 \totalN U_\mu(0) - \Big( \frac{1}{2}m_1^2 + \frac{1}{2}m_1 + cm_1 \Big)\log (\totalN \rho(0)) 
        \\
        &\quad + \Big(\frac{1}{2}(m_1+c)^2 -\frac{1}{12} \Big)\log(m_1+c) - \frac{3}{4}(m_1+c)^2 +\frac{\log(2\pi)}{2}(m_1+c)
        \\
        &\quad +\zeta'(-1) - \log G(c+1) + O(m_1^{-2}) + O(N^{-\frac{1}{2}}m_1^{\frac{5}{2}}(\log N)^3).
\end{split}
\end{align} 
\end{lem}
\begin{proof}
It follows from Lemma \ref{lem:logh1} that
   \begin{equation*}
       \begin{split}
           &\sum_{j=0}^{m_1-1} \log h_{j,N} 
            = \sum_{j=0}^{m_1-1} \Big( 2\totalN U_{\mu}(0) - (j+c+1)\log(\totalN\rho(0))+\log \Gamma(j+c+1)+\varepsilon_{N,2}(j)\Big)\\
           & = 2m_1 \totalN U_{\mu}(0) - \Big(\frac{1}{2}m_1(m_1-1)+ (c+1)m_1\Big)\log(\totalN\rho(0)) + \log\Big(\frac{G(m_1+c+1)}{G(c+1)}\Big) + O(N^{-\frac{1}{2}}m_1^{\frac{5}{2}}(\log N)^3).
       \end{split}
   \end{equation*}
Then by using the asymptotic expansion of the Barnes $G$-function \eqref{Barnes G asymp},  we obtain the desired result.
\end{proof}

\begin{lem}\label{lem:sum3}
   As $N\to \infty$, we have
   \begin{align}
    \begin{split}
     \sum_{j=N-m_2}^{N-1} \log h_{j,N} &= -\Big(\frac{1}{2}m_2^2 +\frac{1}{2}m_2+\alpha m_2 \Big)\log (n\tilde{\rho}(0)) +
    \Big(\frac{1}{2}(m_2+\alpha)^2 -\frac{1}{12}\Big)\log(m_2+\alpha)
    \\
    &\quad - \frac{3}{4}(m_2+\alpha)^2 
    + \frac{\log(2\pi)}{2}(m_2+\alpha) + \zeta'(-1) - \log G(\alpha+1) 
    \\
    &\quad + O(m_2^{-2})+O(N^{-\frac{1}{2}}m_2^{\frac{5}{2}}(\log N)^3).
    \end{split}
   \end{align}
\end{lem}
\begin{proof}
  By Lemma \ref{lem:logh3}, we have
   \begin{equation*}
       \begin{split}
           \sum_{j=N-m_2}^{N-1} \log h_{j,N} 
           & = \sum_{j=N-m_2}^{N-1} 
            \Big ( -(N-j+\alpha)\log(\totalN \tilde{\rho}(0)) +\log\Gamma(N-j+\alpha)+ \varepsilon_{N,3}(j) \Big) \\
            & = - \Big(\frac{1}{2}m_2^2 +\Big(\alpha+\frac{1}{2}\Big) m_2 \Big)\log(\totalN\tilde{\rho}(0))+ \log\Big(\frac{G(m_2+\alpha+1)}{G(\alpha+1)}\Big)+ O(N^{-\frac{1}{2}}m_2^{\frac{5}{2}}(\log N)^3).
       \end{split}
   \end{equation*}
Then again, the asymptotic behaviour \eqref{Barnes G asymp} of the Barnes-$G$ function completes the proof. 
\end{proof}

We are now ready to complete the proof of Theorem~\ref{Thm_free energy det}. 

\begin{proof}[Proof of Theorem~\ref{Thm_free energy det}] 
We choose $m_1 = m_2 = N^{\frac{1}{6}}$ so that all error terms 
$$O(m_j^{-1}(\log N)^d), \qquad O\big(N^{-\frac{1}{2}}m_j^{\frac{5}{2}}), \qquad  O(m_j^{-2}), \qquad O(N^{-\frac{1}{2}}m_j^{\frac{5}{2}}(\log N)^3),\qquad (j=1,2)$$
in Lemmas \ref{lem:sum1}, \ref{lem:sum2}, and \ref{lem:sum3} are $O(N^{-\frac{1}{12}}(\log N)^d)$.
Combining Lemmas \ref{lem:sum1}, \ref{lem:sum2}, and \ref{lem:sum3} with the asymptotic expansion for $t_{m_1}, t_{N-m_2}$ (Lemma~\ref{lem:tj}) and the numerology \eqref{eq: tau} for $\tau_c(m_1), \tau_\alpha(m_2)$, it follows that 
\begin{equation*}
    \begin{split}
        \sum_{j=0}^{N-1}\log h_{j,N} &= -n^2 I[\mu] -\frac{1}{2} N\log n  
        -\frac{n}{2}E[\mu]-2ncU_{\mu}(0)
        +\frac{N}{2}\log 2\pi +
        \frac{1}{2}\Big(\alpha^2+c^2 
        -\frac{1}{3}\Big)\log n\\
        &\quad +\frac{1}{2}\Big(c^2+c+\frac{1}{3}\Big)\log \rho(0)+\frac{1}{2}\Big(\alpha^2 + \alpha + \frac{1}{3}\Big)\log \tilde{\rho}(0) + \frac{1}{2}(\alpha+c)\log 2\pi \\
        &\quad - \frac{5}{12} - \frac{1}{6}\int_{0}^{\infty} \Big(\frac{\rho''(t)}{\rho(t)}-\frac{5}{4}\Big(\frac{\rho'(t)}{\rho(t)}\Big)^2\Big)t\,dt
        +2\zeta'(-1)-\log (G(c+1)G(\alpha+1))+O(N^{-\frac{1}{12}}(\log N)^d)
    \end{split}
\end{equation*}
for some $d>0$. 
To see this, we first note that the terms involving $m_1$ in $\sum_{j=0}^{N-1}\log h_{j,N}$ can be simplified as 
\begin{align*}
\frac{1}{2}m_1\log{n} &-\frac12 \Big((m_1+c)^2-\frac{1}{6}\Big)\log\frac{m_1+c+1/2}{n\rho(0)} + \frac{3}{4}\Big( m_1+c+\frac12\Big)^2 - \frac{1}{2}\Big( m_1+c+\frac12\Big)  (1-\log\rho(0)) \\ 
&- \Big( \frac{1}{2}m_1^2 + \frac{1}{2}m_1 + cm_1 \Big)\log (\totalN \rho(0))  + \Big(\frac{1}{2}(m_1+c)^2 -\frac{1}{12} \Big)\log(m_1+c) - \frac{3}{4}(m_1+c)^2
\end{align*}
up to error terms $O(N^{-\frac{1}{12}}(\log N)^d)$. 
Rearranging terms, we find that the coefficient of $\log n$ in the above is $\frac12 c^2 - \frac1{12}$ and the coefficient of $\log \rho(0)$ is $\frac{1}{2}(c^2+c+\frac{1}{3}).$
Using Taylor expansion of the logarithm function, the remaining terms
\begin{align*}
-\frac12 \Big((m_1+c)^2-\frac{1}{6}\Big)\log\Big(1+ \frac{1}{2(m_1+c)}\Big) + \frac{3}{4}\Big( m_1+c+\frac12\Big)^2 - \frac{1}{2}\Big( m_1+c+\frac12\Big) - \frac{3}{4}(m_1+c)^2
\end{align*}
turn out to be $O(m_1^{-1})$.
A similar asymptotic analysis can be applied to the terms involving $m_2$ in $\sum_{j=0}^{N-1}\log h_{j,N}$.
Then the theorem follows from \eqref{ZN random normal symplectic} and \eqref{log N!}. 
\end{proof}

\subsection{Proof of Theorem~\ref{Thm_free energy Pfaff}}

In this subsection, we prove Theorem~\ref{Thm_free energy Pfaff}. By \eqref{ZN random normal symplectic}, the proof parallels its determinantal counterpart, except that we only need to account for the orthogonal norms of odd degrees. Therefore, we keep the presentation brief to avoid unnecessary repetition.

As a counterpart of Lemma~\ref{lem:sum1}, we have the following. 

\begin{lem}  \label{Lem_4.7}
    As $N \to \infty$, we have 
    \begin{align}
    \begin{split} \label{eq log h sum symp case2}
         &\quad \sum_{j=m_1}^{N-m_2-1}\log h_{2j+1,2N} 
         \\
            &= -2\totalN^2 I[\mu] + \frac{1}{2}(N-m_1-m_2)\log\Big(\frac{\pi}{\totalN}\Big) - \frac{1}{2}\totalN E[\mu] - n(4m_1+4c+1)U_{\mu}(0)\\
            & \quad -\Big( \frac{(4m_1+4c+1)^2}{8} -\frac{5}{24} \Big)\log \tilde{t}_{2m_1+1} + \Big(\frac{(4m_2+4\alpha+1)^2}{8} - \frac{5}{24} \Big)\log \tilde{t}_{2(N-m_2)+1} \\
            &\quad + \frac{3}{2}\totalN^2 \big(\tilde{\tau}_c(2m_1+1)^2 + \tilde{\tau}_\alpha(2m_2)^2\big) - \frac{\totalN}{2}\tilde{\tau}_c(2m_1+1)(2-\log\rho(0))  +  \frac{\totalN}{2}\tilde{\tau}_{\alpha}(2m_2)\log\tilde{\rho}(0)
            \\
            &\quad  + \frac{1}{8}\log\Big(\frac{\tilde{\rho}(0)
            }{\rho(0)}\Big)  - \frac{5}{24} - \frac{1}{12}\int_{0}^{\infty}\Big( \frac{\rho''(t)}{\rho(t)}-\frac{5}{4}\Big(\frac{\rho'(t)}{\rho(t)}\Big)^2 \Big)t\,dt \\
            &\quad + O((m_1^{-1}+m_2^{-1})(\log N)^d) + O(N^{-\frac{1}{2}}(m_1^{\frac{5}{2}}+m_2^{\frac{5}{2}})),
    \end{split}
    \end{align}
\end{lem}

\begin{rem}
In contrast to Lemma~\ref{lem:sum1}, one might initially think that the right-hand side of \eqref{eq log h sum symp case2} (particularly the third line) appears asymmetric with respect to $c$ and $\alpha$.
Nevertheless, it can indeed be shown that they exhibit symmetry in the large-$N$ asymptotic expansion.
This symmetry becomes evident when tracking the asymptotic expansion of the terms involving involving $\tilde{\tau}_c(2m_1+1)$ and $\tilde{\tau}_\alpha(2m_2)$. 
Using Lemma~\ref{lem:tildetj}, one can compare the asymptotic behaviours of the terms 
$$
-\Big( \frac{(4m_1+4c+1)^2}{8} -\frac{5}{24} \Big) \frac12 \log \tilde{\tau}_c(2m_1+1)  + \frac{3}{2}\totalN^2 \tilde{\tau}_c(2m_1+1)^2 - n \tilde{\tau}_c(2m_1+1) 
$$
and 
$$
- \Big(\frac{(4m_2+4\alpha+1)^2}{8} - \frac{5}{24} \Big)\frac12 \log \tilde{\tau}_\alpha(2m_2)  + \frac{3}{2}\totalN^2 \ \tilde{\tau}_\alpha(2m_2)^2  .
$$
Then, it can be verified that the coefficient of $m_1$ in the first equation is $2c + 1/2$, whereas the coefficient of $m_2$ in the second equation is $2\alpha + 1/2$.
Furthermore, these terms cancel out with those in Lemma~\ref{Lem_4.9} below.
\end{rem}
 
\begin{proof}[Proof of Lemma~\ref{Lem_4.7}]
Recall the definition \eqref{eq: symplectic V} of $\widetilde{V},$ the symplectic counterpart of $V$ and the definition \eqref{eq: symplectic critical point} of critical points. 
By using the Euler-Maclaurin formula \eqref{EMfor}, we have  
\begin{equation}
\begin{split}
    \sum_{j=m_1}^{N-m_2-1}\widetilde{V}_{{2j+1}}(\tilde{t}_{2j+1}) 
    & = \totalN \int_{\tilde{t}_{2m_1+1}}^{\tilde{t}_{2(N-m_2)+1}}(-2U_{\mu}(t) - 2(1-F(t))\log t)\cdot 2t \rho(t)\,dt 
    \\
    &\quad + \frac{1}{2}\left(\widetilde{V}_{2m_1+1}(\tilde{t}_{2m_1+1}) - \widetilde{V}_{2(N-m_2)+1}(\tilde{t}_{2(N-m_2)+1})\right) 
    \\
    &\quad + \frac{1}{12}\Big(\d_x (\widetilde{V}_{2x+1}(\tilde{t}_{2x+1}))\big|_{x=N-m_2}-\d_x (\widetilde{V}_{2x+1}(\tilde{t}_{2x+1}))\big|_{x=m_1} \Big) + \tilde{\epsilon}_N,
    \end{split}
\end{equation}  
where $\tilde{\epsilon}_N = O(N^{-3}t_{m_1}^{-4})$.
Adjusting the proof of Lemma~\ref{lem:vsum1} to the symplectic setting, we have  
\begin{align*}
-n  \int_{\tilde{t}_{2m_1+1}}^{\tilde{t}_{2(N-m_2)+1}} (2U_{\mu}(t) + 2(1-F(t))\log t)\cdot 2t\rho(t)\,dt  =  n\Big( I[\mu] + \tilde{\mathcal{E}}_{1,1}(m_1)+ \tilde{\mathcal{E}}_{1,2}(m_2) \Big)+O(N^{-\frac{3}{2}}(m_1^{\frac{5}{2}}+m_2^{\frac{5}{2}})), 
\end{align*} 
as $N\to \infty$, where 
    \begin{equation}\label{def:tcale}
    \begin{split}
        \tilde{\mathcal{E}}_{1,1}(m_1)
        &=2\tilde{\tau}_c(2m_1+1)U_{\mu}(0) - \frac{3}{4}\tilde{\tau}_c(2m_1+1)^{2} + \tilde{\tau}_c(2m_1+1)^2\log \tilde{t}_{2m_1+1} , 
        \\
        \tilde{\mathcal{E}}_{1,2}(m_2) &=-\frac{3}{4} \tilde{\tau}_{\alpha}(2m_2)^2 - \tilde{\tau}_{\alpha}(2m_2)^2 \log \tilde{t}_{2(N-m_2)+1}.
    \end{split}
\end{equation}
Also, by a similar method used in the proof of Lemma \ref{lem:vsum2}, we have 
\begin{align*}
 &\quad \frac{1}{2} \big(\widetilde{V}_{2m_1+1}(\tilde{t}_{2m_1+1})- \widetilde{V}_{2(N-m_2)+1}(\tilde{t}_{2(N-m_2)+1})\big)  
 \\
&= -U_{\mu}(0) + \frac{1}{2}\tilde{\tau}_c(2m_1+1)\Big( 1- 2\log \tilde{t}_{2m_1+1} \Big)      -\frac{1}{2}\tilde{\tau}_\alpha(2m_2) \Big ( 1+2 \log \tilde{t}_{2(N-m_2)+1} \Big)  + O(N^{-\frac{3}{2}}(m_1^{\frac{3}{2}}+m_2^{\frac{3}{2}}))
\end{align*}
and 
\begin{align*}
\frac{1}{12}\big(\d_x (\widetilde{V}_{2x+1}(\tilde{t}_{2x+1}))\big|_{x=N-m_2} - \d_x(\widetilde{V}_{2x+1}(\tilde{t}_{2x+1}))\big|_{x=m_1}\big) =-\frac{1}{6\totalN}\left(\log \tilde{t}_{2(N-m_2)+1} - \log \tilde{t}_{2m_1+1} \right)
\end{align*}
as $N \to \infty.$ 
Furthermore, by repeating the computations in Lemmas \ref{lem:vsum3} and \ref{lem:vsum4}, we have 
 \begin{equation*}
        \begin{split}
            \sum_{j=m_1}^{N-m_2-1}\frac{1}{2}\log\rho(\tilde{t}_{2j+1}) 
            &= \frac{1}{2}nE[\mu] - \frac{1}{2}\totalN \tilde{\tau}_c(2m_1+1)\log \rho(0) - \frac{1}{2}\totalN \tilde{\tau}_\alpha (2m_2)\left(-2+\log\tilde{\rho}(0)-4\log\tilde{t}_{2(N-m_2)+1}\right)\\
            &\quad + \frac{1}{4}\left(\log\rho(\tilde{t}_{2m_1+1}) - \log\rho(\tilde{t}_{2(N-m_2)+1})\right) + O(m_1^{-1}+m_2^{-1})+O(N^{-\frac{1}{2}}(m_1^{\frac{3}{2}}+m_2^{\frac{3}{2}}))
        \end{split}
    \end{equation*}
and 
    \begin{equation*}
        \begin{split}
            \sum_{j=m_1}^{N-m_2-1}\frac{1}{2\totalN}\mathfrak{B}_1(\tilde{t}_{2j+1})&= - \frac{1}{24}\log \Big(\frac{\tilde{t}_{2(N-m_2)+1}}{\tilde{t}_{2m_1+1}}\Big)-\frac{1}{8}\log\Big(\frac{\rho(\tilde{t}_{2(N-m_2)+1})}{\rho(\tilde{t}_{2m_1+1})}\Big)-\frac{5}{24} \\
            &\quad - \frac{1}{12}\int_{0}^{\infty}\Big( \frac{\rho''(t)}{\rho(t)}-\frac{5}{4}\Big(\frac{\rho'(t)}{\rho(t)}\Big)^2 \Big)t\,dt +O(m_1^{-1}+m_2^{-1}). 
        \end{split}
    \end{equation*}
Combining all of the above with Lemma \ref{lem:logh2}, after rearranging and simplifying terms, we obtain the desired result. 
\end{proof}

As a counterpart of Lemmas~\ref{lem:sum2} and ~\ref{lem:sum3}, we have the following. 

\begin{lem} \label{Lem_4.9}
As $ N\to \infty$, we have 
\begin{equation}
    \begin{split}
        \sum_{j=0}^{m_1-1} \log h_{2j+1,2N} 
        &= 4m_1\totalN U_{\mu}(0) - m_1(m_1+2c+1)\log(2\totalN \rho(0))+m_1(m_1+2c)\log 2 - \frac{1}{2}m_1\log \pi \\
        &\quad -\log \Big( G(c+1)G(c+3/2) \Big) + m_1^2\log m_1 - \frac{3}{2}m_1^2 + \Big(2c+\frac{1}{2}\Big)m_1\log m_1 
        \\
        &\quad + \log(2\pi) m_1 - \Big(2c+\frac{1}{2}\Big)m_1 + \frac{1}{2}\Big(2c^2+c-\frac{1}{12}\Big)\log m_1
        \\
        &\quad +\Big(c+\frac{1}{4}\Big)\log 2\pi + 2\zeta'(-1) +O(m_1^{-2}) + O(N^{-\frac{1}{2}}m_1^{\frac{5}{2}}(\log N)^3),
    \end{split}
\end{equation}
    and  
\begin{equation}
    \begin{split}
        \sum_{j=N-m_2}^{N-1}\log{h}_{2j+1,2N} &= -m_2(m_2+2\alpha+1)\log(2\totalN \tilde{\rho}(0))+ m_2(m_2+2\alpha)\log 2 - \frac{1}{2}m_2\log \pi \\
        &\quad -\log \Big(  G(\alpha+1)G(\alpha+3/2)  \Big)  +  m_2^2\log m_2 - \frac{3}{2}m_2^2 + \Big(2\alpha+\frac{1}{2}\Big)m_2\log m_2 
        \\
        &\quad+ \log(2\pi) m_2 - \Big(2\alpha+\frac{1}{2}\Big)m_2 + \frac{1}{2}\Big(2\alpha^2+\alpha-\frac{1}{12}\Big)\log m_2
        \\
        &\quad +\Big(\alpha+\frac{1}{4}\Big)\log 2\pi + 2\zeta'(-1) +O(m_2^{-2})+O(N^{-\frac{1}{2}}m_2^{\frac{5}{2}}(\log N)^3).
    \end{split}
\end{equation} 
\end{lem} 
\begin{proof}
Modifying similar computations used in Lemmas~\ref{lem:sum2} and ~\ref{lem:sum3}, we have  
\begin{equation*}
    \begin{split}
        \sum_{j=0}^{m_1-1} \log h_{2j+1,2N} 
        &= 4m_1\totalN U_{\mu}(0) - m_1(m_1+2c+1)\log(2\totalN \rho(0))+m_1(m_1+2c)\log 2 - \frac{1}{2}m_1\log \pi \\
        &\quad +\log \Big(\frac{G(m_1+c+1)G(m_1+c+3/2)}{G(c+1)G(c+3/2)}  \Big) + O(N^{-\frac{1}{2}}m_1^{\frac{5}{2}}(\log N)^3),
    \end{split}
\end{equation*}
    and \begin{equation*}
    \begin{split}
        \sum_{j=N-m_2}^{N-1}\log{h}_{2j+1,2N} &= -m_2(m_2+2\alpha+1)\log(2\totalN \tilde{\rho}(0))+ m_2(m_2+2\alpha)\log 2 - \frac{1}{2}m_2\log \pi \\
        &\quad +\log \Big( \frac{G(m_2+\alpha+1)G(m_2+\alpha+3/2)}{G(\alpha+1)G(\alpha+3/2)} \Big) +O(N^{-\frac{1}{2}}m_2^{\frac{5}{2}}(\log N)^3)
    \end{split}
\end{equation*}
as $N\to\infty.$
By using \eqref{Barnes G asymp}, we have 
\begin{equation*}
    \begin{split}
        &\quad \log (G(m_1+c+1)G(m_1+c+3/2)) 
        \\
        &= m_1^2\log m_1 - \frac{3}{2}m_1^2 + \Big(2c+\frac{1}{2}\Big)m_1\log m_1 + \log(2\pi) m_1 - \Big(2c+\frac{1}{2}\Big)m_1 + \frac{1}{2}\Big(2c^2+c-\frac{1}{12}\Big)\log m_1\\
        &\quad +\Big(c+\frac{1}{4}\Big)\log 2\pi + 2\zeta'(-1) +O(m_1^{-2})
    \end{split}
\end{equation*}
as $N \to \infty$, and similar asymptotic formula holds for $\log (G(m_2+\alpha+1)G(m_2+\alpha+3/2))$.
\end{proof}

\begin{proof}[Proof of Theorem~\ref{Thm_free energy Pfaff}]
By combining the previous lemmas and straightforward computations, including the cancellations of the terms $m_1$ and $m_2$ as in the proof of Theorem~\ref{Thm_free energy det} in the previous section, we obtain
\begin{equation*}
    \begin{split}
        \sum_{j=0}^{N-1}h_{2j+1,2N} &= -2n^2 I[\mu] - \frac{N}{2}\log \frac{n}{\pi} - \frac{n}{2}E[\mu] - n(4c+1) U_{\mu}(0) + \Big(c^2 + \alpha^2 + \frac{c}{2}+\frac{\alpha}{2}-\frac{1}{12}\Big)\log n \\
        &\quad  + \Big(\alpha+c+\frac{1}{2}\Big)\log (2\pi) + \Big(c^2+c+\frac{5}{24}\Big)\log \rho(0) + \Big(\alpha^2 + \alpha + \frac{5}{24}\Big)\log \tilde{\rho}(0) \\
        &\quad -\frac{5}{24}- \frac{1}{12}\int_0^{\infty} \Big(\frac{\rho''(t)}{\rho(t)} - \frac{5}{4}\Big(\frac{\rho'(t)}{\rho(t)} \Big)^2 \Big)t\,dt + 4\zeta'(-1) \\
        &\quad - \log\big(G(c+1)G(c+\tfrac{3}{2})G(\alpha+1)G(\alpha+\tfrac{3}{2})\big) + O(N^{-\frac{1}{12}}(\log N)^d),
    \end{split}
\end{equation*}
as $N\to \infty$. Now, the theorem follows from \eqref{ZN random normal symplectic} and \eqref{log N!}. 
\end{proof}

\subsection*{Acknowledgements} Sung-Soo Byun was supported by the POSCO TJ Park Foundation (POSCO Science Fellowship), by the New Faculty Startup Fund at Seoul National University and by the LAMP Program of the National Research Foundation of Korea (NRF) grant funded by the Ministry of Education (No. RS-2023-00301976). Nam-Gyu Kang was supported by Samsung Science and Technology Foundation (SSTF-BA1401-51), a KIAS Individual Grant (MG058103) at Korea Institute for Advanced Study, and the National Research Foundation of Korea (RS-2019-NR040050).  Seong-Mi Seo was supported by the National Research Foundation of Korea (NRF-2022R1I1A1A01072052). Meng Yang was supported by the Research Grant RCXMA23007, and the Start-up funding YJKY230037 at Great Bay University.


\begin{thebibliography}{999}

\bibitem{Ad18} K. Adhikari, \emph{Hole probabilities for $\beta$-ensembles and determinantal point processes in the complex plane}, Electron. J. Probab. \textbf{23} (2018), 1--21.

\bibitem{ADM24} G. Akemann, M.~Duits and L. D.~Molag, \emph{Fluctuations in various regimes of non-Hermiticity and a holographic principle}, arXiv:2412.15854. 

\bibitem{AEP22} G. Akemann, M. Ebke and I. Parra, \emph{Skew-orthogonal polynomials in the complex plane and their Bergman-like kernels}, Comm. Math. Phys. \textbf{389} (2022), 621--659.

\bibitem{ABN24} G. Akemann, S.-S. Byun and K. Noda, \emph{Pfaffian structure of the eigenvector overlap for the symplectic Ginibre ensemble}, arXiv:2407.17935.  


\bibitem{ACC23}  Y. Ameur, C. Charlier and J. Cronvall, \emph{Free energy and fluctuations in the random normal matrix model with spectral gaps}, arXiv:2312.13904. 


\bibitem{ACCL23} Y. Ameur, C. Charlier, J. Cronvall and J. Lenells, \emph{Exponential moments for disk counting statistics at the hard edge of random normal matrices}, J. Spectr. Theory \textbf{13} (2023), 841--902.

\bibitem{ACCL24} Y. Ameur, C. Charlier, J. Cronvall and J. Lenells, \emph{Disk counting statistics near hard edges of random normal matrices: the multi-component regime}, Adv. Math. \textbf{441} (2024), 109549.

\bibitem{Barnes} E. W. Barnes, \emph{The theory of the G-function}, Quart. J. Math. \textbf{31} (1899) 264--314.


\bibitem{BBNY19} R. Bauerschmidt, P. Bourgade, M. Nikula, and H.-T. Yau, \emph{The two-dimensional Coulomb plasma: quasi-free approximation and central limit theorem}, Adv. Theor. Math. Phys. \textbf{23} (2019), 841--1002.

\bibitem{BGNW21} R. Butez, D. García-Zelada, A. Nishry and A. Wennman, \emph{Universality for outliers in weakly confined Coulomb-type systems}, arXiv:2104.03959. 


\bibitem{BF23a} S.-S. Byun and P. J. Forrester, \emph{Spherical induced ensembles with symplectic symmetry}, SIGMA Symmetry Integrability Geom. Methods Appl. 19 (2023), 033, 28pp.

\bibitem{BF24} S.-S.~Byun and P. J.~Forrester, \emph{Progress on the study of the Ginibre ensembles}, KIAS Springer Ser. Math. \textbf{3} Springer, 2024, 221pp. 

\bibitem{BFL25} S.-S. Byun, P. J. Forrester and S. Lahiry, \emph{Properties of the one-component Coulomb gas on a sphere with two macroscopic external charges}, arXiv:2501.05061.
 

\bibitem{BKS23} S.-S. Byun, N.-G. Kang and S.-M. Seo, \emph{Partition functions of determinantal and Pfaffian Coulomb gases with radially symmetric potentials}, Comm. Math. Phys. \textbf{401} (2023), 1627--1663.

\bibitem{BP24} S.-S. Byun and S. Park, \emph{Large gap probabilities of complex and symplectic spherical ensembles with point charges}, arXiv:2405.00386.

\bibitem{BSY24} S.-S Byun, S.-M. Seo and M. Yang, \emph{Free energy expansions of a conditional GinUE and large deviations of the smallest eigenvalue of the LUE}, arXiv:2402.18983.

\bibitem{CFTW15} T. Can, P. Forrester, G. Téllez and P. Wiegmann, \emph{Exact and asymptotic features of the edge density profile for the one component plasma in two dimensions}, J. Stat. Phys. \textbf{158} (2015), 1147--1180.

\bibitem{CGJ20} D. Chafaï, D. García-Zelada and P. Jung. \emph{Macroscopic and edge behavior of a planar jellium}, J. Math. Phys. \textbf{61} (2020), 033304.
 
\bibitem{Ch22} C. Charlier, \emph{Asymptotics of determinants with a rotation-invariant weight and discontinuities along circles}, Adv. Math. \textbf{408} (2022), 108600.

\bibitem{Ch23} C. Charlier, \emph{Large gap asymptotics on annuli in the random normal matrix model}, Math. Ann. \textbf{388} (2024), 3529--3587.

\bibitem{Ch23a} C. Charlier, \emph{Hole probabilities and balayage of measures for planar Coulomb gases}, arXiv:2311.15285.  


\bibitem{CK22} J. G. Criado del Rey and A. B. J. Kuijlaars, \emph{A vector equilibrium problem for symmetrically located point charges on a sphere}, Constr. Approx. \textbf{55} (2022), 775--827.

\bibitem{DS22} A. Deaño and N. Simm, \emph{Characteristic polynomials of complex random matrices and Painlevé transcendents}, Int. Math. Res. Not. \textbf{2022} (2022), 210--264.

\bibitem{FK14} F. Ferrari and S. Klevtsov, \emph{FQHE on curved backgrounds, free fields and large $N$}, J. High Energy Phys. \textbf{2014} (2014), 086.  

\bibitem{FKZ12} F. Ferrari, S. Klevtsov and S. Zelditch, \emph{Simple matrix models for random Bergman metrics}, J. Stat. Mech. \textbf{2012} (2012), P04012. 


\bibitem{FF11} J. Fischmann and P. J. Forrester, \emph{One-component plasma on a spherical annulus and a random matrix ensemble}, J. Stat. Mech. Theory Exp. \textbf{2011} (2011), P10003.


\bibitem{Fo92} P. J. Forrester, \emph{Some statistical properties of the eigenvalues of complex random matrices}, Phys. Lett. A \textbf{169} (1992), 21--24.
  
\bibitem{Fo10}
P. J. Forrester, \emph{Log-gases and random matrices}, Princeton University Press, Princeton, NJ, 2010.

\bibitem{Fo13a} P. J. Forrester, \emph{Skew orthogonal polynomials for the real and quaternion real {G}inibre ensembles and generalizations}, J.~Phys. A \textbf{46} (2013), 245203.
 

\bibitem{FK09} P. J. Forrester and M. Krishnapur, \emph{Derivation of an eigenvalue probability density function relating to the Poincaré disk}, J. Phys. A \textbf{42} (2009), 385204.

\bibitem{FM12} P. J. Forrester and A. Mays, \emph{Pfaffian point process for the Gaussian real generalised eigenvalue problem}, Probab. Theory Relat. Fields, \textbf{154} (2012), 1--47.


\bibitem{HW21} H. Hedenmalm and A. Wennman, \emph{Planar orthogonal polynomials and boundary universality in the random normal matrix model}, Acta Math. \textbf{227} (2021), 309--406.

\bibitem{JMP94}  B. Jancovici, G. Manificat and C. Pisani, \emph{Coulomb systems seen as critical systems: finite-size effects in two dimensions}, J. Stat. Phys. \textbf{76} (1999), 307--329.

\bibitem{Kan02} E. Kanzieper, \emph{Eigenvalue correlations in non-Hermitean symplectic random matrices}, J. Phys. A \textbf{35} (2002), 6631--6644.

\bibitem{KKL24} M. Kieburg, A. B. J. Kuijlaars and S. Lahiry, \emph{Orthogonal polynomials in the normal matrix model with two insertions}, arXiv:2408.12952. 

\bibitem{Kl14} S. Klevtsov, \emph{Random normal matrices, Bergman kernel and projective embeddings}, J. High Energy Phys. \textbf{133} (2014), no. 1, 18 pp.



\bibitem{KMMW17} S. Klevtsov, X. Ma, G.  Marinescu and P. Wiegmann, \emph{Quantum Hall effect and Quillen metric}, Comm. Math. Phys. \textbf{349} (2017), 819--855.

\bibitem{Kr09} M. Krishnapur, \emph{From random matrices to random analytic functions}, Ann. Probab. \textbf{37} (2009), 314--346.

\bibitem{LS17} T. Leblé and S. Serfaty, \emph{Large deviation principle for empirical fields of log and Riesz gases}, Invent. Math. \textbf{210} (2017), 645--757.


\bibitem{LY23} S.-Y. Lee and M. Yang, \emph{Strong asymptotics of planar orthogonal polynomials: Gaussian weight perturbed by finite number of point charges}, Comm. Pure Appl. Math. \textbf{76} (2023), 2888--2956.

\bibitem{LD21} A. Legg and P. Dragnev, \emph{Logarithmic equilibrium on the sphere in the presence of multiple point charges}, Constr. Approx. \textbf{54} (2021), 237--257.
 
\bibitem{Lu00} Z. Lu, \emph{On the lower order terms of the asymptotic expansion of Zelditch}, Amer. J. Math. \textbf{122} (2000), 235--273. 


\bibitem{MM12} X. Ma and G. Marinescu, \emph{Berezin-Toeplitz quantization on K\"{a}hler manifolds}, J. Reine Angew. Math. \textbf{662} (2012), 1--56.

\bibitem{May13} A. Mays, \emph{A real quaternion spherical ensemble of random matrices}, J. Stat. Phys. \textbf{153} (2013), 48--69.

\bibitem{MP17}  A. Mays and A. Ponsaing, \emph{An induced real quaternion spherical ensemble of random matrices}, Random Matrices Theory Appl. \textbf{6} (2017), 1750001.


\bibitem{NIST}  F. W. J. Olver, D. W. Lozier, R. F. Boisvert, and C. W. Clark, eds. \emph{NIST Handbook of Mathematical Functions}, Cambridge: Cambridge University Press, 2010.




\bibitem{ST97} E. B. Saff and V. Totik, \emph{Logarithmic Potentials with External Fields}, Grundlehren der Mathematischen Wissenschaften, Springer-Verlag, Berlin, 1997.

 
\bibitem{Se23} S. Serfaty, \emph{Gaussian fluctuations and free energy expansion for Coulomb gases at any temperature}, Ann. Inst. Henri Poincaré Probab. Stat. \textbf{59} (2023), 1074–1142.

\bibitem{Se24} S. Serfaty, \emph{Lectures on Coulomb and Riesz Gases}, arXiv:2407.21194.

\bibitem{TF99} G. Téllez and P. J. Forrester, \emph{Exact finite-size study of the 2D OCP at $\Gamma=4$ and $\Gamma= 6$}, J. Stat. Phys. \textbf{97} (1999), 489--521.

\bibitem{Tian90} G. Tian, \emph{On a set of polarized K\"{a}hler metrics on algebraic manifolds}, J. Diff. Geom. \textbf{32} (1990), 99--130.

\bibitem{WW19} C. Webb and M. D. Wong, \emph{On the moments of the characteristic polynomial of a Ginibre random matrix}, Proc. Lond. Math. Soc. \textbf{118} (2019), 1017--1056.

\bibitem{Zel98} S. Zelditch, \emph{Szeg\"{o} kernels and a theorem of Tian}, Int. Math. Res. Notices \textbf{1998} (1998) 317--331.

\bibitem{ZW06} A. Zabrodin and P. Wiegmann, \emph{Large-$N$ expansion for the 2D Dyson gas}, J. Phys. A \textbf{39} (2006), 8933--8964.
 
 \end{thebibliography}
\end{document}